%% file: paper.tex
\documentclass{sig-alternate}

\usepackage{amsmath}
\usepackage{mathptmx}
\usepackage{scrextend}
\usepackage{balance}  
\usepackage{verbatim}
\usepackage{mathrsfs}
\usepackage{mathtools}
\usepackage{hyperref}
\usepackage[usenames,dvipsnames]{color}
\usepackage{multirow}
\usepackage{textcomp}  
\usepackage{graphicx}
\usepackage{caption}
\usepackage{subcaption}
\usepackage{enumitem}
\usepackage{booktabs}
\usepackage{xspace}
\usepackage{outlines}
\usepackage{wasysym}
\usepackage{algorithm}
\usepackage[noend]{algorithmic}

\numberwithin{equation}{section}	

\newtheorem{thm}{Theorem}[section]

\newtheorem{defn}[thm]{Definition}
\newtheorem{example}[thm]{Example}
\newtheorem{prob}[thm]{Problem}

\newcommand{\Query}{Q}
\newcommand{\QSet}{\mathbf{\Query}}
\newcommand{\DAG}{G_\QSet}
\newcommand{\Conf}{\phi}
\newcommand{\ConfSet}{\Phi}
\newcommand{\pdf}{f_T(t)}

\newcommand{\Dataset}{D}

\newcommand{\Price}{\pi}
\newcommand{\StatsD}{s_\Dataset}

\newcommand{\Cost}{C}
\newcommand{\Contract}{\mathscr{C}}
\newcommand{\Util}{U}
\newcommand{\Itv}{I}
\newcommand{\Demand}{M}
\newcommand{\ProfitAll}{\mathscr{P}}

\newcommand{\Loss}{L}
\newcommand{\Risk}{R}
\newcommand{\Rate}{\alpha_\Cost}
\newcommand{\Ratej}{\alpha_{\Cost j}}
\newcommand{\ParaDemandT}{{\alpha_\Demand}}
\newcommand{\ParaDemandP}{{\beta_\Demand}}
\newcommand{\ParaDemandU}{{\lambda_\Demand}}
\newcommand{\ParaDemand}{{\gamma_\Demand}}
\newcommand{\ParaUtilT}{{\alpha_\Util}}
\newcommand{\ParaUtilP}{{\beta_\Util}}

\newcommand{\Target}{\mathcal{T}}
\newcommand{\target}{\tau}
\newcommand{\Expect}{\hat{T}}
\newcommand{\expect}{\hat{t}}

\newcommand{\manyBucket}{Contract\xspace}

\begin{document}

\title{A Consumer-Centric Market for Database Computation in the Cloud}

\numberofauthors{1}
\author{
	\phantom{aa}Yue Wang\phantom{aaa}\qquad Alexandra Meliou\qquad Gerome Miklau \and \\
	\alignauthor
	    \affaddr{College of Information and Computer Sciences}\\
	    \affaddr{University of Massachusetts Amherst\\}
	    \affaddr{ \{yuewang,ameli,miklau\}@cs.umass.edu}
}

\makeatletter
\def\@copyrightspace{\relax}
\maketitle

\begin{abstract}

The availability of public computing resources in the cloud has revolutionized
data analysis, but requesting cloud resources often involves complex decisions
for consumers. Under the current pricing mechanisms, cloud service providers
offer several service options and charge consumers based on the resources they
use. Before they can decide which cloud resources to request, consumers have
to estimate the completion time and cost of their computational tasks for
different service options and possibly for different service providers. This
estimation is challenging even for expert cloud users.

We propose a new market-based framework for pricing computational tasks in the
cloud. Our framework introduces an \emph{agent} between consumers and cloud
providers. The agent takes data and computational tasks from users, estimates
time and cost for evaluating the tasks, and returns to consumers
\emph{contracts} that specify the price and completion time. Our framework can
be applied directly to existing cloud markets without altering the way cloud
providers offer and price services. In addition, it simplifies cloud use for
consumers by allowing them to compare contracts, rather than choose resources
directly. We present design, analytical, and algorithmic contributions focusing
on pricing computation contracts, analyzing their properties, and optimizing
them in complex workflows.
We conduct an experimental evaluation of our market framework over a
real-world cloud service and demonstrate empirically that our market
ensures three key properties: (a) that consumers benefit from using
the market due to competitiveness among agents, (b) that agents have an
incentive to price contracts fairly, and (c) that inaccuracies in
estimates do not pose a significant risk to agents' profits. Finally,
we present a fine-grained pricing mechanism for complex workflows and
show that it can increase agent profits by more than an order of
magnitude in some cases.

\end{abstract}

\input{intro}

\input{overview}

\input{problem}

\input{pricing}

\input{contract_optimizer}

\input{baselines}

\input{evaluation}

\smallskip

\input{related}

\smallskip

\input{discussion}

\section{Conclusions} \label{sec:conclusion}

In this paper, we propose a new marketplace framework that consumers can use
to pay for well-defined database computations in the cloud. In contrast with
existing pricing mechanisms, which are largely resource-centric, our framework
introduces agent services that can leverage a plethora of existing tools for
time, cost estimation, and scheduling, to provide consumers with personalized
cloud-pricing contracts targeting a specific computational task. Agents
price contracts to maximize the utility offered to consumers while
also producing a profit for their services. Our market can operate in
conjunction with existing cloud markets, as it does not alter the way cloud
providers offer and price services. It simplifies cloud use for consumers by
allowing them to compare contracts, rather than choose resources directly. The
market also allows users to extract more value from public cloud resources,
achieving cheaper and faster query processing than naive configurations, while
a portion of this value is earned by the agents as profit for their services.
Our experimental evaluation using the AWS cloud computing platform
demonstrated that our market framework offers incentives to consumers, who can
execute their tasks more cost-effectively, and to agents, who make profit from
providing fair and competitive contracts.

\balance
\bibliographystyle{abbrv}
\bibliography{reference} 

\input{appendix}

\end{document}

%% file: intro.tex

\section{Introduction} \label{sec:intro}

The availability of public computing resources in the cloud has revolutionized
data analysis. Users no longer need to purchase and maintain dedicated
hardware to perform large-scale computing tasks. Instead, they can execute
their tasks in the cloud with the appealing opportunity to pay for just
what they need. They can choose virtual machines with a wide variety of
computational capabilities, they can easily form large clusters of virtual
machines to parallelize their tasks, and they can use software that is already
installed and configured.

Yet, taking advantage of this newly-available computing infrastructure
often requires significant expertise. The common pricing mechanism of the
public cloud requires that users think about low-level resources (e.g. memory,
number of cores, CPU speed, IO rates) and how those resources will translate
into efficiency of the user's task. Ultimately, users with a well-defined
computational task in mind care most about two key factors: the task's
completion time and its financial cost. Unfortunately, many users lack the
sophistication to navigate the complex options available in the cloud and
to choose a configuration\footnote{A configuration here means a set of system resources and its settings, provided by the cloud provider. It includes the number of virtual instances of a cluster, the buffer size of a cloud database, and so on.}
that meets their preferences.

As a simple example, imagine users who need to execute a workload of
relational queries using the Amazon Relational Database Service (RDS).
They need to select a machine type from a list of more than 20 possible
options, including ``db.m3.xlarge'' (4 virtual CPUs, 15GB
of memory, costing \$0.370 per hour) and ``db.r3.xlarge'' 
(4 virtual CPUs, 30.5GB of memory, costing \$0.475 per hour). The query
workload may run more quickly using db.r3.xlarge, because it has more memory,
however the hourly rate of db.r3.xlarge is also more expensive, which may
result in higher overall cost. Which machine type should the users choose if
they are interested in the cheapest execution? Which machine type should they
choose if they are interested in the cheapest execution completing within 10
minutes? Typical users do not have enough information to make this choice, as
they are often not familiar with configuration parameters or cost models.

The reality of users' choices is even more complex since they may
choose one of five data management systems through RDS, or other query engines using EC2, including parallel processing engines, and different configuration options for each. They might also be tempted to compare multiple service providers, in which case they would have to deal with different pricing mechanisms in addition to different configuration options. Amazon RDS 
charges based on the capacity and number of computational nodes per hour; Google BigQuery 
charges based on the size of data processed; Microsoft Azure SQL Database 
charges based on the capacities of service tiers like database size limit and transaction rate.

As a result of this complexity, many users of public cloud resources make na\"ive, suboptimal choices that result in overpayment, 
and/or performance that is contrary to their preferences (e.g., it exceeds their desired deadline or exceeds their budget). Thus, instead of paying only for what they need, the reality is that they pay for what they do not need and, even worse, they pay more than they have to for it.

\paragraph{A market for database computations}\smallskip
To ease the burden on users we propose a new market-based framework for pricing computational tasks in the cloud. Our framework introduces an entity called an \emph{agent}, who acts as a broker between consumers and cloud service providers. The agent accepts data and computational tasks from users, estimates the time and cost for evaluating the tasks, and returns to consumers \emph{contracts} that specify the price and
completion time for each task.

Our market can operate in conjunction with existing cloud markets, as it does not alter the way cloud providers offer and price services. It simplifies
cloud use for consumers by allowing them to compare contracts, rather than choose resources directly. The market also allows users to extract more value from public cloud resources, achieving cheaper and faster query processing than naive configurations. At the same time, a portion of the value an agent helps extract from the cloud will be earned by the agent as profit.

Agents are conceptually distinct from cloud service providers in the sense that they have their intelligent models to estimate time and cost given consumers queries. In other words, agents take the risk of estimation, while service providers simply charge based on resource consumption, which guarantees profit. In practice, an agent could be a service provider (who provides estimation as a service in addition to cloud resources), a piece of software sold to consumers, or a separate third party who provides service across multiple providers.

\smallskip
\noindent
\textbf{Scope.}
Our goal in this paper is \emph{not} to develop a new technical approach for
estimating completion time or deriving an optimal configuration for a cloud-based computation.  Prior work has considered these challenges, but, in our view, has not provided a suitable solution to the complexity of cloud provisioning. The reason is that estimation, even for relatively well-defined tasks like relational workloads, is difficult.  Proposed methods require complicated profiling tasks to generate models and specialize to one type of workload (e.g., Relational database \cite{Karampaglis:2014} or MapReduce \cite{Herodotou:2011:VLDB}).  In addition, there is inherent uncertainty in prediction, caused by multi-tenancy common in the cloud~\cite{Wang:2010:HotCloud, Schad:2010, Farley:2012, Tudoran:2012, Kiefer:2014}. Lastly, users' preferences are complex, involving both completion time  \cite{Ramakrishnan:2011} and cost \cite{Pandey:2010, Yi:2012, Kokkinos:2013, Zhou:2014, Li:2015}, which have been considered as separate goals~\cite{Herodotou:2011:SOCC, Malawski:2012, Marcus:2016}, but have not been successfully integrated.  

Our market-based framework incentivizes expert agents to employ combinations
of existing estimation techniques to provide this functionality as a service
to non-expert consumers. Users can express preferences in terms of their
\emph{utility}, which includes both time and cost considerations. Uncertainty
in prediction becomes a risk managed by agents, and included in the price of
contracts, rather than a problem for users. Ultimately our work complements
research into better cost estimation in the cloud~\cite{Xiong:2011, Chi:2013,
Herodotou:2011:SOCC}. In fact, our market will function more effectively
as such research advances and agents can exploit new techniques for better
estimation.

Our work makes several contributions:

\begin{itemize}[leftmargin=5mm, itemsep=0mm, topsep=0mm]

\item We define a novel market for database computations,
including flexible contracts reflecting user preferences.  

\item We formalize the agent's task of pricing contracts and propose an efficient algorithm for optimizing contracts. 

\item We perform extensive evaluation on Amazon's public cloud, using benchmark queries and real-world scientific workflows.  
We show that our market is practical and effective, and satisfies key properties ensuring that both consumers and agents benefit from the market.

\end{itemize}

The paper is organized as follows. We present an overview of the market and
main actors in Section~\ref{sec:overview}. We formally define contracts and
optimal pricing of contracts in Section~\ref{sec:contracts} and \ref{sec:pricing}. 
We extend our framework to support fine-grained pricing to further optimize contracts in Section~\ref{sec:fineGrained}. In Section~\ref{sec:baseline}, we introduce several alternatives.
In Section~\ref{sec:evaluation}, we
present a thorough evaluation of our proposed market, and demonstrate that it guarantees several important properties. Finally, we discuss related work and extension and summarize our contributions
in Sections~\ref{sec:related}, \ref{sec:discussion} and~\ref{sec:conclusion}, respectively.

%% file: overview.tex

\section{Computation market overview} \label{sec:overview}

In this section, we discuss the high-level architectural components of our
computation market: three types of participants and their interactions through
computation contracts. Our computation market exhibits several desirable
properties, which we mention in Section~\ref{sec:prop_overview}.

\subsection{Market participants} \label{sec:participants}

Our goal is to model the interactions that occur in a computation market, and
design the roles and framework in a way that ensures that the market functions
effectively. Our computation market involves three types of participants:

\begin{figure}[t]
\begin{center}
\includegraphics[width=.95\columnwidth]{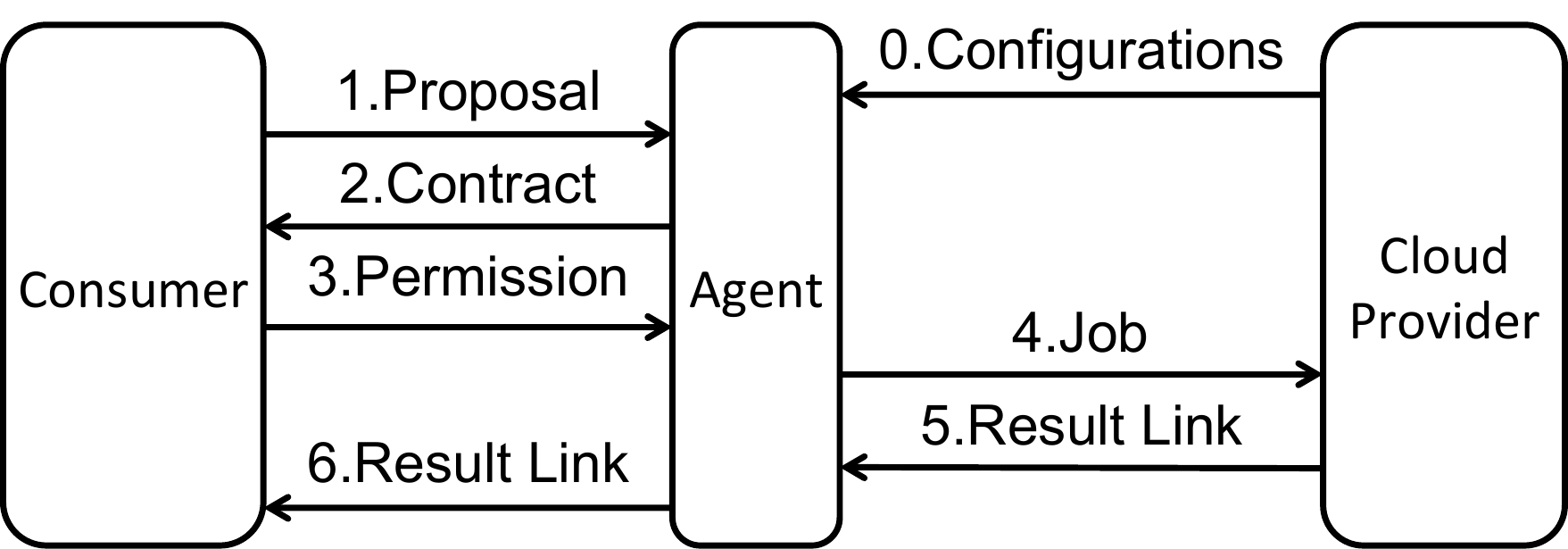}

\caption{An overview of interactions of the main participants in the computation market: the consumer, the agent, and the cloud provider.}
\label{fig:framework}
\end{center}
\end{figure}

\begin{description}[leftmargin=5mm]

\item[Cloud provider.] 
Cloud providers are public entities that offer computational resources as a
service, on a pay-as-you-go basis.  These resources are often presented as virtual machine
types and providers charge fees based on the capabilities of the virtual
machines and the duration of their use. Our framework does not enforce any
assumptions on the types, quantity, or quality of resources that a cloud
provider offers.

\item[Consumer.] \looseness -1
A consumer is a participant in our computational market who needs to complete
a computational task over a dataset $\Dataset$. We assume the computational
task is a set of queries or MapReduce jobs\footnote{For simplicity of terminology we use ``query'' to refer to either a query or MapReduce job.}, denoted as $\QSet = \{Q_1, \allowbreak Q_2, ..., \allowbreak Q_n\}$.
We assume that the consumer does not own the computational resources needed to
complete $\QSet$, and thus needs to use cloud resources. However, the consumer
may not have the expertise to determine which cloud provider to use, which
resources to lease, or how to configure them. In our framework, the consumer wishes to retrieve the task
results $\QSet(\Dataset) = \{Q_1(\Dataset), \allowbreak Q_2(\Dataset),
\allowbreak ... , \allowbreak Q_n(\Dataset)\}$ within a specified timeframe,
and pay for these results directly. Therefore, the consumer's goal is to
complete the task efficiently and for low cost. Different consumers have different time and cost preferences. They will describe these preferences precisely using a utility function, as described later in Section \ref{sec:utility}.

\item[Agent.] \looseness -1
Consumers' needs are task-centric (time and price to complete a given task),
whereas cloud providers' abilities are resource-centric (time and price for a
type of resource). Due to this disparity, consumers and providers do not
interact directly in our framework. Rather, a semantically separate entity,
the \emph{agent}, is tasked with handling the interactions between consumers
and cloud providers. The agent receives a task request from a consumer and, in
response, calculates a price to complete the task, providing the consumer with a
formal contract. We review contracts in Section~\ref{sec:contract_overview},
and describe them in detail in Section~\ref{sec:contracts}.
The agent executes accepted contracts using public cloud resources, and earns
a profit whenever the contract price is greater than the actual cost of executing the
task. The agent's goals are to attract business by pricing contracts competitively
and to earn a profit with each transaction. One of the main challenges for the
agent is to assign accurate prices to consumer requests, which requires
knowledge of cloud resources, their capabilities and costs, and
expertise in tuning and query prediction.
\end{description}

Figure \ref{fig:framework} illustrates the interactions among the three market
participants. 
In step 0, the agent collects details on available
configurations from the cloud provider to derive later price quotes on consumers' requests. 
This step may only need to be initiated once, and reused afterwards.
In steps 1 through 3, the agent receives a proposal including $\QSet$ and statistics
about dataset $\Dataset$, denoted $\StatsD$, which are sufficient for pricing. 
For example, $\StatsD$ can be the number of input records in each table \cite{Akdere:2012}, histograms on key columns or sets of columns \cite{Wu:2013:ICDE, Wu:2013:VLDB, Wu:2014}, a small sample of data \cite{Herodotou:2011:VLDB}, and other standard statistics relevant to the task.
The agent reasons about possible configurations and estimates the completion time and
financial cost of the queries, returning a priced contract to the consumer. If
the consumer accepts the contract, in steps 4 through 6, the agent executes a
job in the cloud according to the contract, computes the result, and returns a
link to the consumer. The link can be, for example, an URL pointing to Amazon
S3 or any other cloud storage service. Finally the agent receives payment
based on the accepted contract and the actual completion time.  We will see in Section~\ref{sec:contractDefinition} that contracts can involve complex prices that depend on the actual completion time.

\subsection{Contracts}\label{sec:contract_overview}

The contract is the core component of our framework, describing the terms of a
computational task the agents will perform and the price they will receive upon
completion of the task. The design of our market framework is intended to cope with the inevitable uncertainty of completion time.  Therefore, our contracts support variable pricing based on the actual completion time when the answer is delivered.  

We also formally model the time/cost preferences of the consumer using a \emph{utility function} that we assume is shared with the agent. The main technical challenge for the agent is to price a contract of interest to a consumer.  Pricing relies on the agent's model of expected completion time for the task as well as the consumer's utility.  From the consumers' side, they may receive and compare contracts from multiple agents in order to
choose the one that maximizes their utility. 

In this paper, we consider contracts and computational tasks that only
involve analytic workloads. These analytic workloads are different
from long-running services in the sense that their evaluation takes
limited amount of time, even though this time can be several hours or
days. Given this focus, we can assume that cloud resources do not
change during the execution of task. This means, for example, that the
capacity of virtual machines and their rate remain the same during the
execution of a contract. We discuss relaxing these factors in
Section~\ref{sec:discussion}.

\subsection{Properties and assumptions}\label{sec:prop_overview}

Our framework is designed to support three important properties:
competitiveness, fairness, and resilience. \emph{Competitiveness} guarantees
that agents have an incentive to reduce runtime and/or cost for consumers. 
\emph{Fairness}
guarantees that agents have an incentive to present accurate estimates to
consumers, and that they do not benefit by lying about expected completion
times. \emph{Resilience} means that an agent can profit in the marketplace
even when their estimates of completion time are imprecise and possibly
erroneous. We demonstrate empirically in Section~\ref{sec:evaluation} that our framework satisfies these crucial properties.

Our framework assumes honest participants; we defer the study of malicious consumers and agents to future work. Accepting an agent's contract means the consumer's data will be shared with the agent for evaluation of their task, however requesting contract prices from a set of agents reveals only the consumer's statistics and task description. 

Monopoly is not possible in this framework, and collusion among agents
is unlikely.\footnote{In fact, the agents and the existing cloud
service providers naturally form a monopolistic competition
\cite{varian2010intermediate, Mankiw2014}.} First, an agent cannot
constitute a monopoly, since consumers may always choose to use a
cloud service provider directly. A service provider cannot constitute
a monopoly either, as any agent with a valid estimation model can
enter the market. Second, collusion becomes unlikely as the number of
agents in the market increases. Any agent who does not collude with
others can offer a lower price and draw consumers, making any
collusion unstable.

%% file: problem.tex

\section{The consumer's point-of-view} \label{sec:contracts}

In this section, we describe the consumer's interactions with the market. A
transaction begins with a consumer who submits a request. This request
reflects their \emph{utility}, which is a precise description of
their preferences. Later, given multiple priced contracts, the consumers can
formally evaluate them according to the likely utility they will offer.

\subsection{Consumer utility} \label{sec:utility}

One of our goals is to avoid simplistic definitions of contracts in which a
task is carried out by a deadline for a single price. For one, many consumers
have preferences far more complex than individual deadlines: they can tolerate
a range of completion times, assuming they are priced appropriately. In
addition, we want agents to compete to offer contracts that best meet the
preferences of consumers.

A consumer's preferences are somewhat complex because they involve tradeoffs
between both completion time and price. We adopt the standard economic notion
of consumer \emph{utility}~\cite{varian2010intermediate} and model it
explicitly in our framework.
A utility function precisely describes a consumer's preferences by associating
a utility value with every (time, price) pair. A utility function can encode,
e.g., the fact that the consumer is indifferent to receiving their query
answer in 10 minutes at a cost of \$2.30 or 20 minutes at a cost of \$1.90
(when these two cases have equal utility values) or that receiving an
answer in 30 minutes at a cost of \$0.75 is preferable to both of the above
(when it has strictly greater utility).

\begin{defn}[Utility] 
\emph{Utility} $\Util(t, \Price)$ is a real-valued function of time and price,
which measures consumer satisfaction when a task is evaluated in time $t$ with
price $\Price$. 
\end{defn}

Larger values for $\Util(t, \Price)$ mean greater utility and a preferred
setting of $t$ and $\Price$. For a fixed completion time $t_0$, a consumer
always prefers a lower price, so $\Util(t_0, \Price)$ increases as $\Price$
decreases. Similarly, for a fixed price, $\Price_0$, a consumer always prefers
a shorter completion time, so $\Util(t, \Price_0)$ increases with decreasing
$t$.

To simplify the representation of a consumer's utility, we will restrict our
attention to utility functions that are piecewise linear. That is, we assume
the range of completion times $[0,\infty)$ is divided into a fixed set of
intervals, and that utility on each interval is defined by a linear function
of $t$ and $\Price$. This means that for each interval, the consumer has a
(potentially different) rate at which she/he is willing to trade more time for lower price, and vice versa. 

\begin{defn}[Utility -- piecewise] 
    A piece-wise utility function consists of a list of target times $\target_0,
    \dots, \target_n$, where $0 = \target_0 < \target_1 < \dots < \target_{n-1} < \target_n = \infty$, and linear functions $u_1(\pi,t), \dots, u_n(\pi,t)$.
The utility is $u_{i}(\pi,t)$ for $t \in [\target_{i-1},\target_i)$. 
\end{defn}
Such utility functions can express conventional deadlines, but also much more
subtle preferences concerning the completion time and price of a computation.

\begin{example} \label{ex:utility}
Consumer Carol has two target completion times for her computation: 10 minutes
and 20 minutes. Results returned in less than 10 minutes are welcome, but she
doesn't wish to pay more to speed up the task. When results are returned
between 10 minutes and 20 minutes, every minute saved is worth 1 cent to her.
She does not want result returned after 20 minutes. Her piecewise utility
function is:
\begin{equation*}
\Util(t, \Price) =
\left\{
\begin{aligned} 
      u_1(t, \pi) = & -\Price & (t < 10) \\
      u_2(t, \pi) = & -t - \Price + 10 & (10 \leq t < 20) \\
      u_3(t, \pi) = & -50 & (t \geq 20) \\
\end{aligned}
\right.
\end{equation*}
Figure~\ref{fig:utility} depicts $\Util(t, \Price)$ when $t < 20$.
\end{example}

In practice, users can construct the utility function by defining
several critical points on a graphical user interface, or answering a
few simple pair-wise preference questions.

\begin{figure}
        \centering
                \includegraphics[width=0.4\textwidth]{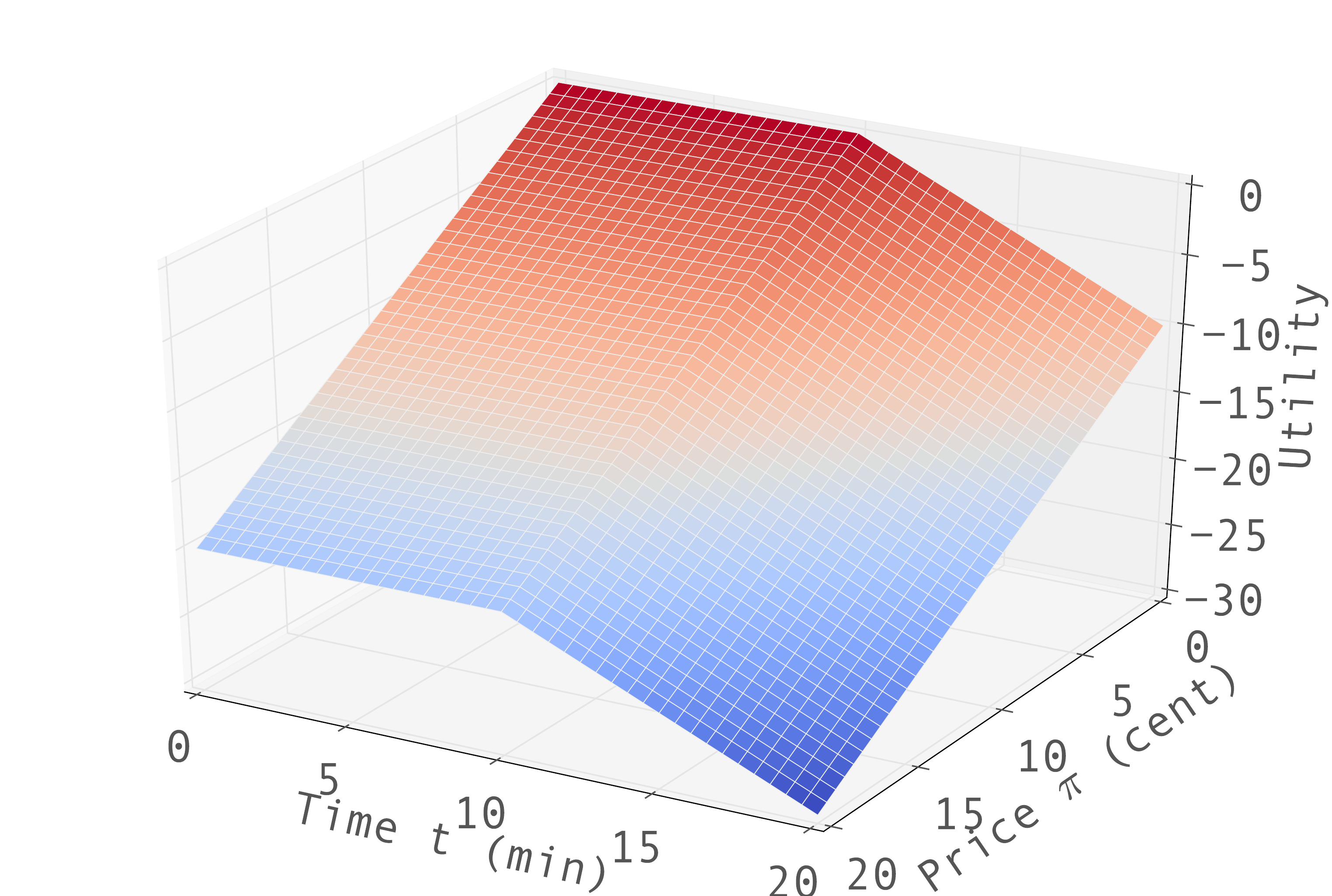}
                \label{fig:utilityTwoDeadline}
        \caption{Utility function for Example \ref{ex:utility} when $t<20$.}\label{fig:utility}
\end{figure}

\subsection{Consumer contract proposal} \label{sec:contractDefinition}

The process of agreeing on a contract starts with the consumer advertising to
agents the basic terms of a contract: the task $\QSet$, the statistics of the
database $\StatsD$, and their piecewise utility function $U$.

The terms of the contract are structured around the target times in the
utility function. Agents use the utility function to choose a suitable
configuration and pricing to match the preferences of the consumer. A
complete, priced contract is returned to the consumer, which is defined as
follows:

\begin{defn}[\manyBucket] \label{def:manyBucket}
A contract is a six-tuple $\Contract=(\QSet, \allowbreak \StatsD, \allowbreak
\Target, \allowbreak P, \allowbreak \Expect, \allowbreak \Pi)$, where $\QSet$ is a task,
$\StatsD$ consists of statistics about the input data, $\Target=(\target_0, \target_1, \dots ,
\target_n)$ is an ordered list of target completion times, $P=(p_1, \dots, p_n)$ is an
ordered list of probabilities, $\sum_i p_i=1$, $\Expect=(\expect_1, \dots, \expect_n)$ is an
ordered list of expected completion times, and $\Pi=(\pi_1(t), \dots, \pi_n(t))$
is a list of price functions where $\pi_i$ is defined on $[\target_{i-1},\target_i)$. 
\end{defn}

When a consumer
and agent agree on a contract $\Contract$, it means that the agent has
promised to deliver the answer to task $\QSet$ on $\Dataset$ after time
$t\in[0,\infty)$, where the likelihood that $t$ falls in interval
$[\target_{i-1},\target_{i})$ is $p_{i}$. Accordingly, if the answer is delivered in the
time interval $[\target_{i-1},\target_{i})$ the consumer agrees to pay the specified
price, $\pi_{i}(t)$. $\Expect$ is used for computing expected utility as we will see in Section~\ref{sec:consumerContractEvaluation}.
The data statistics $s_D$ are given to the agent by the consumer; the agent includes them in the contract because the pricing calculation relies on these statistics.

The contract is an agreement to run the task once. The probabilities provided
by the agent are a claim that if the task were run many times, a fraction of
roughly $p_i$ of the time, the completion time would be in the interval
$[\target_{i-1},\target_{i})$. Without this information, the consumer has no way to
effectively evaluate the alternative completion times that could occur in a
contract. For example, all alternatives but one could be very unlikely and
this would change the meaning of the contract. We will see in
Section~\ref{sec:pricing} how the agent generates these probabilities.

\begin{figure}
        \centering
       
        \begin{subfigure}[b]{0.15\textwidth}
            \includegraphics[width=\textwidth]{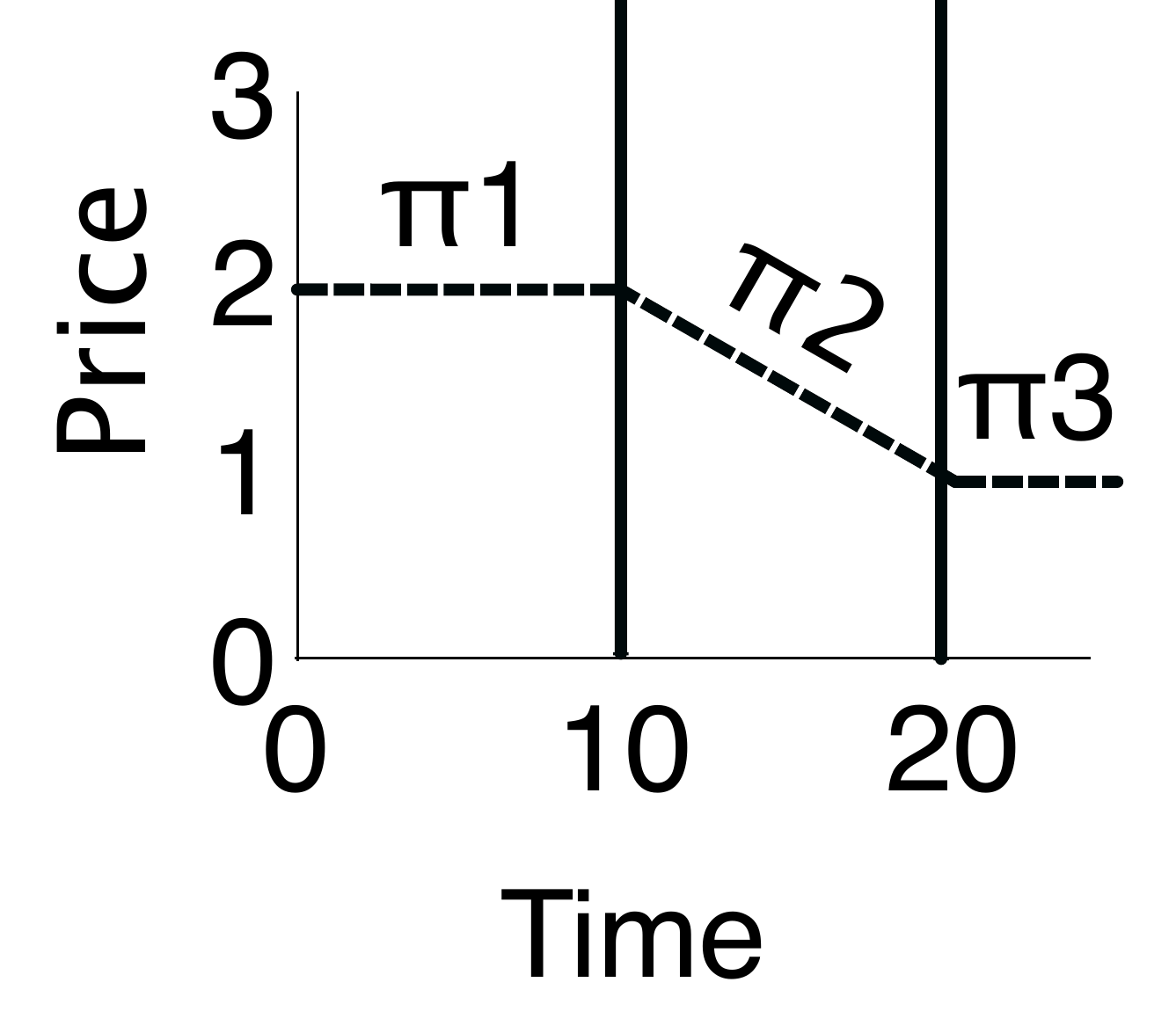}
            \caption{} \label{fig:contracts}
        \end{subfigure}
        \hfill
        \begin{subfigure}[b]{0.31\textwidth}
            \begin{center}
            \begin{tabular}{lcc}
             \toprule
            \textbf{Contracts} & $\mathbf{\Contract_1}$ & $\mathbf{\Contract_2}$ \\
             \midrule
             $\Expect$ & $(9,15,21)$ & $(9,15,21)$\\
             $P$ & $(0.2,0.5,0.3)$ & $(0.1,0.8,0.1)$ \\
             \midrule
             Utility & -18.65 & -10.4 \\
             \bottomrule 
            \end{tabular}
            \end{center}
            \caption{}
            \label{tbl:utilityComsumer}
        \end{subfigure}
        \caption{
        (a) Price function for Example~\ref{ex:priceFunction}.
        (b) Comparison of two contracts.
        }
        \label{fig:contractComparison}
\end{figure}

\begin{example} \label{ex:priceFunction}
An example contract based on the utility function of Example~\ref{ex:utility} is defined by $\Target=(0,10,20,\infty)$, probabilities $P=(0.2,0.5,0.3)$, expectations $\Expect=(9,15,21)$, and prices $\Pi$ (also illustrated in Figure~\ref{fig:contracts}) defined as:
\begin{equation*}
\Pi(t) = 
\left\{
\begin{aligned}
     \pi_1(t) = & \: 2 & (t < 10) \\
      \pi_2(t) = & \: 3 - 0.1 t  & (10 \leq t < 20) \\
      \pi_3(t) = & \: 1  & (t \geq 20) 
\end{aligned}
\right.
\end{equation*}
\end{example}

\subsection{Consumer's contract evaluation} \label{sec:consumerContractEvaluation}

In response to a proposed contract, a consumer hopes to receive a number of priced versions of the contract from agents.  Each contract may offer the consumer a different range of utility values over the probability-weighted completion times.  The consumer's goal is to maximize their utility, so to choose between contracts, the consumer should compute the expected utility of each contract and choose the one with greatest expected utility.
All contracts based on 1 utility request should share the same target completion times.

\begin{defn}[Expected utility of a contract] \label{defn:expectedUtility}
The expected utility of a contract $\Contract=(\QSet, \allowbreak \StatsD, \allowbreak \Target, \allowbreak P, \allowbreak \Expect, \allowbreak \Pi)$ with respect to utility function $\Util(t,\Price)$ is 
$$
\sum_{i=1}^n  p_i  u_i(\expect_i,\pi_i(\expect_i)) 
$$
when $u_i(t, \pi)$ and $\pi_i(t)$ are linear functions. 
\end{defn}

\begin{example}
Suppose the consumer uses the utility function in Example \ref{ex:utility}, and two agents return two contracts $\Contract_1$ and $\Contract_2$. Further assume both agents return the same price function $\Pi$ in Example \ref{ex:priceFunction}, and the expected time $\Expect$ 
are also the same. Only the probabilities $P$ differ as illustrated in Figure \ref{tbl:utilityComsumer}. The consumer computes the expected utility according to Definition \ref{defn:expectedUtility} and chooses $\Contract_2$ as it has greater utility.
\end{example}

%% file: pricing.tex

\section{The agent's point-of-view} \label{sec:pricing}

We now explain the agent's interactions in the market. The agent's main
challenge is to assign prices to a contract, coping with the uncertainty of
completion time, while taking into account the consumer's utility and the
market demand. We formalize two variants of pricing (risk-aware and
risk-agnostic) and formulate both as optimization problems.

\subsection{Pricing preliminaries}

Upon receipt of the terms of a contract and the utility function of a
consumer, the agent must complete the contract by computing prices for each
interval and assigning probabilities to each interval.

For each configuration, we assume the financial cost $\Cost$ borne by the
agent is a function of $t$: $\Cost(t) = \Rate \cdot t$, where $\Rate$ is the
unit rate of the configuration, and can be different across configurations. Thus, the pricing of a contract depends
critically on the estimate of the completion time for $\QSet$. Since estimates
of completion time are uncertain, we model completion time $T$ as a
probability distribution over $[0, \infty)$ with probability density function
$\pdf$. The true $\pdf$ is unlikely to be known and, in practice, must be
estimated by the agent with respect to a selected configuration. Based on
$\pdf$ and $\Cost(t)$, the agent proposes a price function $\Price(t)$,
which means the consumer should pay $\Price(t)$ when the completion time is
$t$.

The agent has three goals when pricing a contract: (i) to maintain
profitability, (ii) to offer the consumer appealing utility, and (iii) to
compete with the offerings of other agents. We discuss each of these goals
below.

\paragraph{(i) Profitability}\smallskip
Naturally the agents would like to price the contract higher than their cost of
execution so that they can earn a profit. Profit is uncertain for an agent
because it is difficult to predict completion time in the cloud. We say a
contract is \emph{profitable in expectation} if its expected profit, with
respect to the distribution $\pdf$, is greater than zero.
\begin{equation}
\label{eqa:expectedProfit}
\begin{aligned}
E[\mbox{\it profit}]  = &\sum_{i=1}^np_i\left(\Price_i(\expect_i)-C(\expect_i)\right) \\
\end{aligned}
\end{equation}
We call a contract profitable (for the agent) as long as it is profitable in
expectation. The agents should always price contracts so that they are
profitable, but it is possible that a particular contract ends up being
unprofitable.

\begin{defn}[Profitable contract] 
A profitable contract is a contract with $E[\mbox{profit}] > 0$.
\end{defn}

\paragraph{(ii) Prioritizing consumer utility}\smallskip
Since the agents knows the consumer's utility function $\Util(t, \Price)$
they can (and should) take it into account when choosing a configuration and
pricing. To the extent that the agents can match the consumer's utility, their
pricing of the contract will be more appealing to the consumer. The agents can
evaluate the expected utility $E[\Util]$ over the distribution of 
time $T$ based on their estimates and price function $\Price(t)$:
\begin{equation} \label{eq:expectedUtility}
E[U]=\sum_{i=1}^n p_iu_i(\expect_i,\Price_i(\expect_i))
\end{equation}

Profitability for the agent and utility for the consumer are conflicting
objectives: a contract that offers greater profit to the agent will offer
lower utility to the consumer. We will see that the agent will attempt to
maximize the consumer's utility, subject to constraints on their
profitability.

\paragraph{(iii) Market competitiveness and demand}\smallskip 
In all markets, including ours, market forces and competition prevent agents
from raising prices without bound. In economics, a market demand function
describes how these forces impact the pricing of
goods~\cite{varian2010intermediate}. 

When the agents decrease the price of a contract, the expected profit of the
contract is reduced but they increase the utility of the contract to
consumers. In a marketplace, when utility for the consumer increases, a
greater number of consumers will accept the contract. Thus, the agents must
balance the profit made from an individual contract with the overall profit
they make from selling more contracts. To model this, we must make an
assumption about the relationship between utility and the number of contracts
that will be accepted by consumers in the market. This relationship is
represented by the \emph{demand function} which is defined as a function of
utility. A linear demand curve is common in
practice~\cite{varian2010intermediate}, so we focus on demand functions of the
form $\Demand(\Util) = a + b\Util$. Our framework can support demand functions
of different forms, but we do not discuss these in detail.

In a real market, agents would learn about demand through repeated
interactions with consumers. An agent's demand function could depend on, for
example, customer loyalty, the best contracts competitors can offer, and
other factors. These factors are beyond our scope. In order to simulate the
functioning of a realistic market, we must assume a demand function and, for
simplicity, we assume the demand functions of all agents are the same in the rest of
this paper.

\subsection{Contract pricing}

We start from the simplest case in which the consumer has a task $\QSet$ and a
single configuration $\Conf$. So the cost function $\Cost(t)$ and the pdf of
the distribution of completion time $\pdf$ are fixed. The agent needs to
define the price function $\Price(t)$ to present a competitive contract to the
consumer. Let the overall profit be $\ProfitAll$, which equals the unit profit
\emph{profit} multiplied by the sales $\Demand(\Util)$. Notice that
\emph{profit} is the profit of a single contract while $\ProfitAll$ is the
overall profit of all contracts that the agent returns to all consumers in the
market. The agent wants to find the price function that leads to the greatest
total profit while satisfying the profitability constraint. This results in
the following optimization problem:

\begin{prob}[Contract pricing] \label{prob:Pricing}
Given a contract $\Contract=(\QSet, \allowbreak \StatsD, \allowbreak \Target,
\allowbreak P, \allowbreak \Expect, \allowbreak \Pi)$, utility function
$\Util$, and demand function $\Demand$, the optimal price for $\Contract$ is:
\begin{equation*}
\begin{aligned}
\mbox{maximize}: &\quad \ProfitAll =  E[\mbox{profit}] \cdot E[\Demand(\Util)] \\
\mbox{subject to}: &\quad E[\mbox{profit}]  > 0\\ 
\end{aligned}
\end{equation*}
\end{prob}

Let $\Itv_i$ be the interval $(t_i, t_{i+1})$, and recall that $p_i$ is the
probability that the completion time falls in $\Itv_i$:
\begin{equation}
\label{eq:discreteP}
p_i = \int_{t=t_i}^{t_{i+1}} f_T(t) dt  
\end{equation}
Let $T_i$ be a random variable of completion time in interval $\Itv_i$. 
It is
a truncated distribution with probability density function $f_T(t | t \in
\Itv_i)$.
Let $C_i$ be a random variable of cost in interval $\Itv_i$.  $C_i = \Cost(T_i)$.
So expectation $\hat{t}_i$ and expectation $c_i$ is:
\begin{gather} \label{eq:discreteT}
\hat{t}_i = E[T_i] = \int_{t \in \Itv_i} t f_T(t | t \in \Itv_i) dt \\
\label{eq:discreteC}
c_i = E[C_i] = \int_{t \in \Itv_i} C(t) f_T(t | t \in \Itv_i) dt 
\end{gather}
Therefore the expected unit profit and expected demand are:
\begin{gather} \label{eq:discreteProfit}
E[\mbox{\it profit}] = \sum_{i=1}^{|I|} (\Price_i - c_i)p_i \\
 \label{eq:discreteDemand}
E[\Demand(\Util)] = \sum_{i=1}^{|I|} \Demand \left( \Util(\expect_i, \Price_i)\right) p_i 
\end{gather}

\paragraph{Linear case}\smallskip
\label{sec:linearCaseSolution}
When $\Util$ and $\Demand$ are linear functions, this problem becomes a convex quadratic programming problem. It has an analytical solution.
More details can be found in the appendix. Here we describe the conclusion only, under the following assumptions:  

\begin{itemize}[leftmargin=5mm]
\item The consumer specifies a linear utility function $\Util(t, \Price) =
-\ParaUtilT \cdot t - \ParaUtilP \cdot \Price$, which means they are always
willing to pay $\ParaUtilT$ units of cost to save $\ParaUtilP$ units of time.

\item The demand function is linear: $\Demand(\Util) = \ParaDemand +
\ParaDemandU \cdot \Util$. Thus, when $\Util$ increased by
$1/\ParaDemandU$, $1$ more contract would be accepted. Since $\Util(t,
\Price)$ is linear, the demand function can be written as $\Demand(\Util) =
\ParaDemand - \ParaDemandT t - \ParaDemandP \Price$.

\end{itemize}

Applying Equations~\ref{eq:discreteProfit} and~\ref{eq:discreteDemand} to Problem~\ref{prob:Pricing}, we compute the overall profit $\ProfitAll$. $\ProfitAll$ is maximized when

{\small
\begin{equation*}
\bold{\Price}^T\bold{p} =
\left\{
\begin{aligned}
\frac{\ParaDemand - \ParaDemandT \bold{\expect}^T\bold{p} + \ParaDemandP \bold{c}^T\bold{p}}{2\ParaDemandP},\; &  \ParaDemand - \ParaDemandT \bold{\expect}^T\bold{p} - \ParaDemandP \bold{c}^T\bold{p} \geq 0 \\
\bold{c}^T\bold{p} + \epsilon,\; & otherwise
\end{aligned}
\right.
\end{equation*}
}
where $\epsilon$ is a small positive value, and 4 vectors $\bold{\Price}=[\Price_1, \Price_2, ...]^T$, $\bold{p}=[p_1, p_2, ...]^T$, $\bold{\hat{t}}=[\hat{t}_1, \hat{t}_2, ...]^T$ and $\bold{c}=[c_1, c_2, ...]^T$.

Furthermore, when $\ParaDemand - \ParaDemandT \bold{\expect}^T\bold{p} - \ParaDemandP \bold{c}^T\bold{p} \geq 0$, 
\begin{equation} \label{eq:linearProfit}
\ProfitAll =  \frac{(\ParaDemand - \ParaDemandT \bold{\expect}^T\bold{p} - \ParaDemandP \bold{c}^T\bold{p})^2}{4 \ParaDemandP}
\end{equation}

\paragraph{Selecting a configuration}\smallskip
An agent typically has many available configurations for evaluating $\QSet$.
We denote the set of configurations by $\ConfSet=\{\Conf_1, \allowbreak \Conf_2,
\allowbreak ...\}$. Every configuration $\Conf_j$ has its own cost function
$\Cost_j(t) = \Ratej \cdot t$, where $\Ratej$ is the unit rate for $\Conf_j$.

The agent will select the configuration that results in the most
profit. The distribution of time $T$ and its corresponding $p_i$, $\expect_i$, and $c_i$ then become variables in
Problem~\ref{prob:Pricing}. A na\"ive agent can select and enumerate a small $\ConfSet$ to find the best possible solution. A smarter agent will use an analytic model to solve the problem \cite{Wu:2014, Herodotou:2011:SOCC}.

\subsection{Risk-aware pricing} \label{sec:riskAware}

Pricing contracts involves some risk for the agents: if their estimated
distributions of time and cost are different from the actual ones, they can
lose profit or even suffer losses.
Next, we formally define risk based on loss and add it as part of the objective.

\begin{defn}[Loss]
Let the actual distribution of completion time be $T^*$ and the optimal price
function be $\Price^*$. When the agent generates a contract with price
function $\Price$, the loss of revenue $\Loss$ is: $\Loss_{T^*}(\Price) =
\ProfitAll(\Price^*, T^*) - \ProfitAll(\Price, T^*) = \ProfitAll(\Price^*,
p^*, t^*, c^*) - \ProfitAll(\Price, p^*, t^*, c^*)$, where $p^*$ is the actual probabilities, $t^*$ is the actual expected completion times, and $c^*$ is the actual costs.
\end{defn}

There is always inherent uncertainty in the prediction of the distributions of
completion of time and cost, so it is generally not possible for the agents to
achieve the theoretically optimal profits based on the actual distributions.
However, they can plan for this risk, and assess how much such risk they are
willing to assume.
We proceed to define risk as the worst-case possible loss that an agent can suffer.
\begin{defn}[Risk]\label{def:risk}
    The risk of the agent is a function of price $\Price$, and is defined as the maximum loss over possible distributions of completion time:
$\Risk(\Price) = \max_{T^*} \Loss_{T^*}(\Price)$.
\end{defn}
We incorporate risk into the agent's optimization problem by adding it to the objective function:
\begin{equation}
\label{eq:riskAware}
\begin{aligned}
\mbox{maximize}:\; & \ProfitAll(\Price, p, t, c) - \lambda\Risk(\Price) \\
\mbox{subject to}:\; & E[profit]  > 0\\ 
\end{aligned}
\end{equation}
The parameter $\lambda$ in the objective is a parameter of risk that the agent
is willing to assume. Larger values of $\lambda$ reduce the worst-case losses
(conservative agent), while smaller values of $\lambda$ increase the assumed
risk (aggressive agent).
The agent can estimate the risk $\Risk(\Price)$ by solving the following
optimization problem, with variables $\Price^*$, $p^*$, $t^*$, and $c^*$:
\begin{equation*}
\begin{aligned}
\mbox{maximize}:\; & \Loss_{T^*}(\Price) = \ProfitAll(\Price^*, p^*, t^*, c^*) - \ProfitAll(\Price, p^*, t^*, c^*) \\
\mbox{subject to}:\; & E[profit^*]  =  \sum(\Price_i^* - c_i^*) p_i^* > 0\\
&LBound_t \leq t^*_i - \hat{t}_i \leq UBound_{t} \\ 
&LBound_c \leq c^*_i - c_i \leq UBound_{c} \\ 
&0 \leq p^*_i  \leq 1 \\ 
&\sum p^*_i = 1 \\ 
\end{aligned}
\end{equation*}
where $LBound$ and $UBound$ are empirical values set by the agent. For instance, an agent's analytic model reports estimated $\hat{t_1}=1$ min. However, the agent has executed 10 contracts and the actual mean of the time is $t_1=1.1$ min. The agent can set $LBound_t = 0$ and $UBound_t = 0.1$.

%% file: contract_optimizer.tex

\section{Fine-grained contract pricing} \label{sec:fineGrained}

Our treatment of pricing in Section~\ref{sec:pricing} assumes that agents
select a single configuration for the execution of a consumer contract.
However, computational tasks often contain well-separated, distinct subtasks
(e.g., operators in a query plan or components in a workflow). These subtasks
may have vastly different resource needs.  For example, Juve et al. \cite{Juve:2013} profile multiple scientific workflow systems including Montage~\cite{Jacob:2009} 
and SIPHT~\cite{Livny:2008} and find that their components have dramatically different I/O, memory and computational requirements.

We now extend our pricing framework to support \emph{fine-grained} pricing, which allows agents to optimally assign separate configurations to each subtask of a computational task.
It provides more candidate contracts without changing the pricing Problem~\ref{prob:Pricing}. 
Fine-grained pricing has two benefits. First, by assigning a configuration for each subtask, instead of the entire task, agents can achieve improved time and cost, resulting in higher overall utility and/or higher profit. Second, considering subtasks separately gives agents the flexibility
to outsource some computation to other agents. While outsourcing computation
across agents is not a focus of our work, it is a natural fit for our
fine-grained pricing mechanism. Agents can choose to outsource subtasks to
other agents based on their specialization and capabilities, or for load
balancing. However, some challenges of outsourcing, such as
utility and forms of contracts that agents need to exchange are beyond the scope of our current work.

We model a computational task $\QSet$ as a directed acyclic graph (DAG)
$\DAG$. Every node in $\DAG$ is a subtask $\Query_i$. An edge between
subtasks $(\Query_i, \Query_j)$ means that the output of $\Query_i$ is an
input to $\Query_j$.  When subtasks are independent of one another, the DAG may be disconnected.  Our model assumes no pipelining in subtask
evaluation. Therefore, a subtask $\Query_j$ cannot be evaluated until all
subtasks $\Query_i$, such that $(\Query_i, \Query_j)\in \DAG$, have
completed their execution.

\begin{figure}[t]
\begin{center}
\includegraphics[width=0.7\columnwidth]{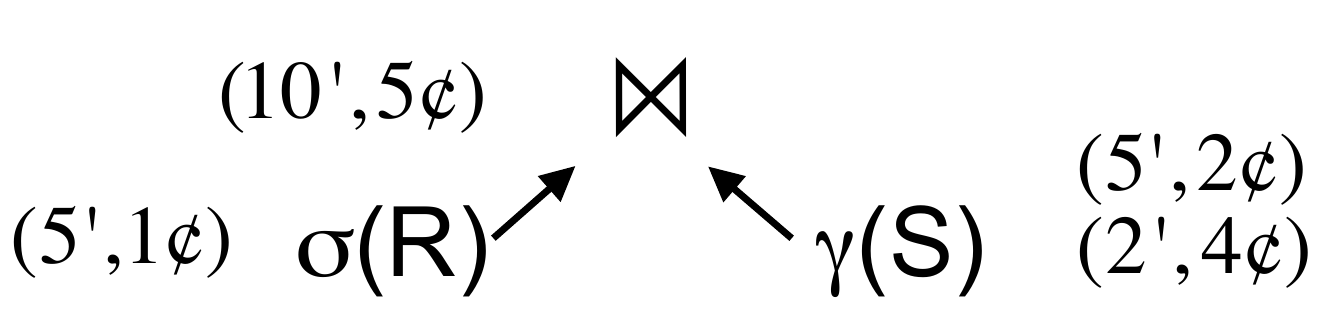}
\caption{An example of a simple relational query that can be broken into 3 subtasks, corresponding to different operators.
} 
\label{fig:contractOptEx}
\end{center}
\end{figure}

Given the graph representation $\DAG$ of a computational task, an agent
needs to determine a configuration $\Conf_i\in\ConfSet$ for each subtask $\Query_i\in
\QSet$. This is in contrast with coarse-grained pricing
(Section~\ref{prob:Pricing}), where the agent had to select a single 
configuration from $\ConfSet$ to be used for each subtask of $\QSet$. 
When the agent chooses $\Conf_i$, the time and cost of $\Query_i$ is $T_i(\Conf_i)$ and $C_i(\Conf_i)$.
A set of selected configurations results in total cost $C_\QSet=\sum_{\Query_i}C_i(\Conf_i)$, i.e., the sum of the costs of
all subtasks. The completion time of $\QSet$ is determined by the
\emph{longest} path ($P$) in the task graph: $T_\QSet=\max_{P\in\DAG}\sum_{\Query_i\in
P}T_i(\Conf_i)$. 
Given demand $\Demand$ and contract utility $\Util$, $T_\QSet$ and
$C_\QSet$ determine the agent's profit $\ProfitAll$. The goal of the agent is to select the set of configurations that maximizes
$\ProfitAll$.

\begin{prob}[Fine-grained contract pricing] \label{prob:Optimize}
Given  graph $\DAG$ representing a task $\QSet$, and possible configurations $\ConfSet$, the agent needs to specify a configuration $\Conf_i\in\ConfSet$ for each $\Query_i \in \QSet$, so that the time $T_\QSet=    \max_{P\in\DAG}\sum_{\Query_i\in
    P}T_i(\Conf_i)$ and cost $C_\QSet=\sum_{\Query_i}C_i(\Conf_i)$ maximize the overall profit $\ProfitAll$.
\end{prob}
Our problem definition does not model data storage and transfer time and costs
explicitly. Rather, we assume that these are incorporated in the time and cost
of a subtask ($T_{\Query_i}$ and $C_{\Query_i}$). 
This
simplifies the model and offers an upper bound on time and cost. In practice,
when two subsequent tasks share the same configuration, it is possible to
reduce the costs of data passing, but these optimizations are beyond the scope
of this work.

We demonstrate the intricacies of the fine-grained pricing problem through a
simple example. Figure~\ref{fig:contractOptEx} shows a relational query with
three distinct subtasks (operators): (1)~select tuples from relation R,
(2)~aggregate on relation S, and (3)~join of the results. We assume
deterministic times and costs to evaluate each subtask, denoted next to each
node in Figure~\ref{fig:contractOptEx}. The select and join subtasks have only
a single possible configuration each, but the aggregate subtask has two.
Assume the utility function is $\Util(t,\Price)=-t-\Price$, which means every
one minute is worth $1$ cent for the consumer. Therefore, the configuration
$(2', 4\cent)$ is better for the aggregate subtask, since it has higher
utility than the configuration $(5',2\cent)$. However, following a greedy
strategy that picks the configuration that is optimal for each subtask can
result in sub-optimal utility for the overall task. In this example, the join
subtask has to wait 5 minutes for the select subtask to complete. Therefore,
there is no benefit in paying a higher price to complete the aggregate subtask
sooner, making $(5',2\cent)$ a better configuration choice.

\begin{thm}
Fine-Grained Contract Pricing is NP-hard.  
\end{thm}
Our reduction follows from the discrete Knapsack problem (Appendix \ref{apx:fineGrained}). We
next introduce a pseudo-polynomial dynamic programming algorithm for this problem, and show that
it is both efficient and effective in real-world task workflows
(Section~\ref{sec:fineExperiments}). Without loss of generality, we assume that time and cost are deterministic, but the algorithm can be extended to the probabilistic case in a straightforward way. 

\begin{algorithm}[t] 

\caption{Fine-Grained Contract Pricing}
\label{alg:fineGrained}
\begin{algorithmic}[1]
    \REQUIRE $\QSet, \DAG, \ConfSet, \ProfitAll(T, C)$
    \ENSURE $maximum~\ProfitAll$
    \STATE Add node $\Query_{terminal}$ with $0$ time and cost to $\DAG$
    \FORALL{$\Query_i \in \QSet$}
      \STATE Add edge $(\Query_i, \Query_{terminal})$ to $\DAG$
    \ENDFOR
    \STATE $\QSet_{order} \Leftarrow TopologicalSort(\DAG)$
    \STATE $boundT \Leftarrow$ longest time to evaluate $\DAG$
    \FORALL{$\Query_i \in \QSet_{order}$}
      \STATE $f(\Query_i, 0) \Leftarrow \infty$ \label{alg:fineGrained_lineFStart}
      \FOR{$t \Leftarrow 1~\TO~boundT$}
        \STATE $f(\Query_i, t) \Leftarrow f(\Query_i, t - 1)$
        \FORALL{$\Conf \in \ConfSet$}
          \STATE $cost_{\Conf} \Leftarrow \underset{ q \in pred(\Query_i)}{Combine}\left(f(q,t - T_i(\Conf))\right) + C_i(\Conf)$ \label{alg:fineGrained_lineCost}
          \IF{$cost_{\Conf} < f(\Query_i, t)$}
            \STATE $f(\Query_i, t) \Leftarrow cost_{\Conf}$
          \ENDIF
        \ENDFOR
      \ENDFOR \label{alg:fineGrained_lineFEnd}
    \ENDFOR
    \RETURN $\max_{t} \ProfitAll(t, f(\Query_{terminal}, t))$
\end{algorithmic}
\end{algorithm}

Algorithm~\ref{alg:fineGrained} uses dynamic programming to compute the optimal profit for task graphs.  The algorithm derives the exact optimal solution for cases where $\DAG$ is a tree (e.g., relational query operators) and computes an approximation of the optimum for task graphs that are DAGs.

In Algorithm~\ref{alg:fineGrained}, $f(\Query_i,t)$ represents the minimum
cost for evaluating the subgraph terminated at subtask $\Query_i$ when it
takes at most time $t$.\footnote{We turn the continuous space of time $t$ into
discrete space by choosing an appropriate granularity (e.g., minute).} 
Then, $f(\Query_i,t)$ can be computed based on a combination of the costs of
the direct predecessors of $\Query_i$ ($pred(\Query_i)$) in the task workflow
(lines~\ref{alg:fineGrained_lineFStart}--\ref{alg:fineGrained_lineFEnd}).
When $\DAG$ is a tree, the $Combine$ function
(line~\ref{alg:fineGrained_lineCost}) is simply the sum of the costs of the
predecessors ($\sum_{q \in pred(\Query_i)}f(q, t-T_i(\Conf))$), and Algorithm~\ref{alg:fineGrained} results in the optimal
profit.

If $\DAG$ is not a tree, predecessors of a subtask $\Query_i$ can share common indirect predecessors, which introduces complex dependencies in the choice of configurations across different subtrees.  For example, let 
$q_1, q_2 \in pred(\Query_i)$, and $q_0\in pred(q_1)\cap pred(q_2)$.  Therefore, $q_0$ affects both subgraphs terminated at $q_1$ and $q_2$, respectively. This impacts the $Combine$ function in two ways.  First, the cost of $q_0$ should be counted only once.  Second, there may be discrepancies in the configuration choice for $q_0$ by the different subgraphs.
There are three strategies to resolve the discrepancy:
(1)~use the configurations with minimum time $T$;
(2)~use the configurations with minimum cost $C$;
(3)~use the configurations with maximum $\ProfitAll(T,C)$.
The $Combine$ function applies the above strategies one by one, computes the time $T_{\Query_i}$ and cost $C_{\Query_i}$ of the subgraph terminated at $\Query_i$, and updates $f(\Query_i, t)$ if $T_{\Query_i} \leq t$ and $C_{\Query_i}$ is better.  Note that Strategy 1 guarantees a feasible solution whenever one exists.

%% file: baselines.tex

\section{Alternative Approaches} \label{sec:baseline}

\paragraph*{Benchmark-Based Approach} ~

Floratou et al.~\cite{Floratou:2011} propose a Benchmark as a Service (BaaS) approach to help consumers select configurations. This new BaaS benchmarks user's workload and use the optimal configuration to execute the workload repetitively. As they mention in the paper, changes such as growth of input data make BaaS complicated. So a BaaS provider need to monitor and react to these changes. In our approach, we do not make assumptions about the repetitiveness of workloads. The disadvantage is that consumers may pay more for repetitive workloads even when they are very similar. The advantage is that consumers do not need to worry about the change of the workloads.

We compare with the benchmark-based approach in Section~\ref{sec:extraEval}.

\paragraph*{VCG-Auction-Based Approach} ~

VCG auction is a pricing mechanism. Its strategy-proof property makes it popular in many studies. We develop a VCG auction model and compare it with our approach. In this VCG model, a customer opens a bidding and agents bid on prices. Notice that this model defines consumers payment according to the utility instead of the pure price, which is different from the canonical VCG mode.

Assume agent $i$ proposes its contract with utility $Util_i$. The consumer takes the contract with the highest utility $Util_i$ but pays based on the second highest utility $\Util^*=\max_{j \neq i}(\Util_j)$. The payment is a piecewise function $\Omega(t)=\omega_k(t)$. From agent i's perspective, its cost function is $\Cost_i(t)$, so its payoff is: 
\begin{equation*}
\text{\emph{payoff}}_i =
\left\{
\begin{aligned} 
      E[\text{\emph{profit}}] = \sum_{k=1}^n p_k(\omega_k(\hat{t_k}) - \Cost_i(\hat{t_k})) & & \text{if}~\Util_i > \max_{j \neq i} (\Util_j) \\
      0 & & otherwise \\
\end{aligned}
\right.
\end{equation*}

Here is the definition of $\Omega(t)$ when $\Util(t,\Price)= -\ParaUtilT t - \ParaUtilP \Price$ is linear: Let $\Delta=\Util_i- \Util^*$. The inverse function of $\Util(t,\Price)$ is $\Pi(t,u)=(-\ParaUtilT t - u)/\ParaUtilP$. We define $\Omega(t)=\Pi(t,u-\Delta)=\Pi(t,u)+\Delta/\ParaUtilP$.

\begin{example}
Suppose the utility function is $\Util(t,\Price)=-t-2\Price~~~(0<t<\infty)$. So $\Pi(t,u)=(-t-u)/2$. When $\Delta=\Util_i-\Util^*=10$, $\Omega(t)=\Pi(t,u-\Delta)=(-t-(u-10))/2=(-t-u)/2+5$.
\end{example}

\begin{thm}
When $\Util(t,\Price)$ is linear, the above $\Omega(t)$ satisfies:

1) $\Util_i > \Util^* > \Util^C_i \Rightarrow E[\text{\emph{profit}}] < 0$;

2) $\Util^C_i > \Util^* > \Util_i \Rightarrow E[\text{\emph{profit}}] > 0$.
\end{thm}
The proof is in Appendix~\ref{apx:vcg}. Given the payment function, we can show that our developed VCG auction is strategyproof.

\begin{thm} 
Every agent truthfully revealing its cost is a weakly-dominant strategy. 
\end{thm}

In other words, every rational agent will make its cost be the price  in its contract. This proof is very similar to the proof for the canonical VCG auction. Please refer to Appendix~\ref{apx:vcg} for detailed proof.

Our posted-price model and VCG auction model both exist in the real world market. We discussed in the related work section that neither of them dominates the other. They have 2 main differences in our case:
1) Posted-price model requires the agent to better understand the demand function of the market. Then an agent can set the price actively to gain more profit. In contrast, agents in VCG auction only needs to truthfully reveal their costs, then the price is passively decided based on the second best utility. 
2) VCG auction requires a centralized auctioneer who ensures the consumers pay according to the second best utility. It makes a cross-platform market more difficult to form. Posted-price model does not have such requirement.

We quantitatively compare with the VCG-auction-based approach in Section~\ref{sec:extraEval}.

\paragraph*{Differentiated Bertrand} ~

Multiple agents competing in the market is a typical differentiated Bertrand model. Specifically, Bertrand model solves the equilibrium of optimal prices based on all agents' demand functions. But in practice, an individual agent can hardly know other agents' demand function when pricing in the market. 
So our model assumes that each agent observes a demand function of its own price. Such demand function is valid given the other agents' current prices. An agent will change its price when others change. The prices will converge to the equilibrium in the long term.
Now we show the connection through an example.

In a differentiated Bertrand model, suppose there are 2 agents numbered 1 and 2. Their prices are $\mu_1$ and $\mu_2$. Agent 1's demand function is $M_1(\mu_1,\mu_2)=\gamma-\alpha \mu_1+\beta \mu_2$ where $\alpha,\beta,\gamma$ are positive parameters. Higher $\mu_1$ means lower $M_1$ but higher $\mu_2$ means higher $M_1$ due to competition. Similarly, $M_2(\mu_1,\mu_2)=\gamma-\alpha\mu_2+\beta\mu_1$. So the overall profit $\ProfitAll_1=M_1(\mu_1,\mu_2)\cdot\mu_1$ \footnote{The Bertrand model in~\cite{rasmusen:2007} assumes marginal cost $c=0$ for simplicity. So the price in~\cite{rasmusen:2007} corresponds to the profit in our approach. i.e. $\mu+c=\Price$.}, 
$\ProfitAll_2=M_2(\mu_1,\mu_2)\cdot\mu_2$. Given a fixed $\mu_1$, the best $\mu_2$ that maximizes $\ProfitAll_2$ is:
\begin{equation} \label{eq:bertrand1}
\mu_2^*=(\gamma+\beta\mu_1)/2\alpha.
\end{equation}
Similarly, the best $\mu_1$ is:
\begin{equation} \label{eq:bertrand2}
\mu_1^*=(\gamma+\beta\mu_2)/2\alpha.
\end{equation}
So one can solve these two equations to get a Nash Equilibrium: $\mu_1^*=\mu_2^*=\gamma/(2\alpha-\beta)$.

In the real world, agents' demand functions cannot be exactly the same. Thus we should use different demand functions $M_1(\mu_1,\mu_2)=\gamma_1-\alpha_1 \mu_1+\beta_1 \mu_2$ and $M_2(\mu_1,\mu_2)=\gamma_2-\alpha_2 \mu_1+\beta_2 \mu_2$.  So $\mu_1^*=(2\alpha_2\gamma_1+\beta_1\gamma_2)/(4\alpha_1\alpha_2-\beta_1\beta_2)$ and $\mu_2^*=(2\alpha_1\gamma_2+\beta_2\gamma_1)/(4\alpha_1\alpha_2-\beta_1\beta_2)$. 

The above calculation, however, requires that both agents know both demand functions $M_1$ and $M_2$, which may not be a realistic assumption; one may easily obtain one's own demand function by fitting historical data, but it may not be possible to know the other party's demand function. Without knowledge of the other party's demand function, one generally cannot settle for the NE $(\mu_1^*, \mu_2^*)$ in one shot, but has to adjust one's price dynamically according to the observed demand function. Thus we have a repeated game here, and rational agents will follow the best response functions~\ref{eq:bertrand1} and \ref{eq:bertrand2}. Our approach uses exactly the same response functions, where the impact of the other agent's price is absorbed into the intercept. More precisely, in our model, Agent 1 tries to optimize $M_1(\mu_1)=\gamma_1'-\alpha_1\mu_1$ where $\gamma_1'=\gamma_1+\beta_1\mu_2$, so it sets $\mu_1=\gamma_1'/2\alpha_1=(\gamma_1+\beta_1\mu_2)/2\alpha_1$. Similarly, Agent 2 sets $\mu_2=(\gamma_2+\beta_2\mu_1)/2\alpha_2$. If both agents keep updating their prices in this manner, their prices will eventually converge to the NE $(\mu_1^*, \mu_2^*)$. 
Please refer to Appendix~\ref{apx:bertrand} for more detail.

%% file: evaluation.tex

\section{Experimental evaluation}
\label{sec:evaluation}

In this section we evaluate our market using a real-world cloud computing platform: Amazon Web Services (AWS).
Our experiments collect real-world data from a variety of relational and
MapReduce task workloads, and use this data to simulate the behavior of our
market entities on the AWS cloud. Our results demonstrate that our market
framework offers incentives to consumers, who can execute their tasks more
cost-effectively, and to agents, who make profit from providing fair and
competitive contracts.

We proceed to describe our experimental setup, including 
computational tasks, consumer parameters, and contracts.

\begin{figure}[tb]
\begin{center}
\begin{tabular}{lcrc}
\toprule
 {\textbf{Type}} & {\textbf{CPU (virtual)}} & \textbf{Memory} & \textbf{\$/hour}  \\ 
 \midrule
 db.m3.Medium & 1 & 3.75GB & \$0.095  \\
 db.m3.Large & 2 &  7.5GB & \$0.195  \\
 db.m3.xLarge & 4 & 15GB & \$0.390  \\
 db.m3.2xLarge & 8 & 30GB & \$0.775  \\
 db.r3.Large & 2 &  15GB & \$0.250  \\
 db.r3.xLarge & 4 & 30.5GB & \$0.500  \\
 db.r3.2xLarge & 8 & 61GB & \$0.995  \\
 m1.Medium & 1 & 3.75GB & \$0.109  \\
 m1.Large & 2 & 7.5GB & \$0.219  \\
 m1.xLarge & 4 & 15GB & \$0.438  \\
 \bottomrule
\end{tabular}
\end{center}
\caption{Types of Amazon machines and associated features and costs (in January 2015). The first 7 types (db.*) are RDS configurations, whereas the last 3 (m1.*) are EMR configurations. The prefixes (db and m1) are omitted from some figures for brevity.}
\label{tbl:amazonMachines}
\end{figure}%

\paragraph{Data and configurations}\smallskip 
We spent 8,106 machine hours and \$3,118 in obtaining the distributions of time and cost for two types of computational tasks: relational query
workloads, and MapReduce jobs. 

\smallskip
\noindent
\textbf{Relational query tasks:} 
We use the queries and data of the TPC-H benchmark 
to evaluate
relational query workloads. We use all 22 queries of the benchmark on a 5GB
dataset (scale factor 5). We use the Amazon Relational Database Service 
(RDS) to evaluate the
workloads on 7 machine configurations,
each of which has 200GB of Provisioned IOPS SSD storage, and runs PostgreSQL
9.3.5. Figure~\ref{tbl:amazonMachines} lists the capacity and hourly rate of
each configuration.

\smallskip
\noindent
\textbf{MapReduce tasks:}
We evaluate MapReduce workloads using three job types (WordCount, Sort, and Join) over 5GB of randomly
generated input data. We use the Amazon Elastic
MapReduce service 
(EMR) to test our framework on these workloads. We select
3 machine configuration types.
Figure~\ref{tbl:amazonMachines} lists the capacities and hourly rates of these
configurations. We experimented with 4 different sizes of clusters for each
machine configuration: 1, 5, 10, and 20 slave nodes.
 
\smallskip
\noindent
\textbf{Scientific workflows:}
We use real-world scientific workflows that represent computational tasks with
multiple subtasks, to evaluate fine-grained pricing
(Section~\ref{sec:fineGrained}). We retrieved 1,454 workflows from
MyExperiment~\cite{DeRoure:2009}, one of the most popular scientific workflow
repositories. These workflows were developed using Taverna
2~\cite{Wolstencroft:2013}, and comprise the majority of workflows in the
repository. The size of workflows ranges from 1 to 154 subtasks.

\paragraph{Consumer models}\smallskip 
We simulate the consumer behavior in our framework using the utility and demand functions.

\smallskip
\noindent
\textbf{Utility:}  
In our evaluation, utility is a linear function $\Util(t, \Price) = -\ParaUtilT
t - \ParaUtilP \Price$ modeling consumer preferences. 
 $\ParaUtilT$ represents the unit cost that the consumer is willing to pay to save
$\ParaUtilP$ unit time. For our experiments, we assume $\Util(t, \Price) = - t -
\Price$, where $t$ is measured in minutes and $\Price$ is measured in cents,
which means every minute is worth 1 cent to the consumer.

\smallskip
\noindent
\textbf{Demand:}
In our evaluation, the demand function is linear: $\Demand(\Util) =
\ParaDemand + \lambda_\Demand \Util$, which means that when $\Util$
increases by $1/\lambda_\Demand$, $1$ more contract would be accepted.
We use $\Demand(\Util) = 100 + 50 \Util$ for RDS,
and $\Demand(\Util) = 100 + 5 \Util$ for EMR. 
$\lambda_\Demand$ is smaller for EMR
because the times and costs for MapReduce jobs are much larger than
those of relational queries.

\paragraph{Contracts}\smallskip
All our experiments involve contracts with a deadline, 
which means that every
consumer request specifies one target completion time. We execute each task 100 times using
every configuration and set the deadline of each query as the average completion time across all configurations.

\subsection{Consumer incentives}

In this section, we evaluate whether our market framework offers sufficient
incentive for consumers to participate in the market. Our first set of
experiments simulates several na\"ive cloud users 
who select one of the
default configurations for their computational tasks: 7 configurations for
RDS, and 12 configurations for EMR.
Then we simulate a baseline user who intuitively chooses configurations based on a simple feature of a task. Specifically, the user chooses a configuration with the best CPU performance for a CPU-intensive task, or a configuration with the best IO performance for an IO-intensive task. Predicting whether a task is CPU-intensive or IO-intensive is a difficult task for a user. However, we unfairly bias toward this baseline by indeed executing the task to measure its CPU time and IO time. We consider a task is CPU-intensive if its CPU time is greater, otherwise IO-intensive.
Finally we
simulate an expert agent, who, for every task, selects the
configuration that maximizes the consumer's utility function.
Figure~\ref{fig:usabilityRDS_UScatter} presents the price and time achieved by
each of the 7 default configurations for RDS, as well as the price and time offered by
the expert agent. The line in the graph is the utility indifference curve for
the agent's configuration, representing points with the same utility value.
Points on the curve are equally good, from the consumer's perspective, as the
one achieved by the expert agent. Points above the curve have worse utility
values (less preferable than the agent's offer), while points below the curve
have better utility values (more preferable than the agent's offer).

Our experiments show that the expert agent provides more utility to 4 out of 8
na\"ive cloud users with relational query tasks on RDS. This means
that, even though the agent makes a profit, a good portion of the users
would still benefit from using the market instead of relying on default
settings. This effect is even more pronounced for EMR workloads.
Figure~\ref{fig:usabilityEMR_UScatter} shows that the expert agent offers
better utility to \emph{all} simulated na\"ive users. This means that, in
every single case, the consumers get better utility by using the agent's
services, instead of selecting a default configuration.
It is noteworthy that the heuristic-based baseline approach is 189\% and 67\% worse in utility than our approach for RDS and EMR workloads, respectively.

\begin{figure*}[t]
        \centering        
        \begin{subfigure}[b]{0.33\textwidth}
                \includegraphics[width=\textwidth]{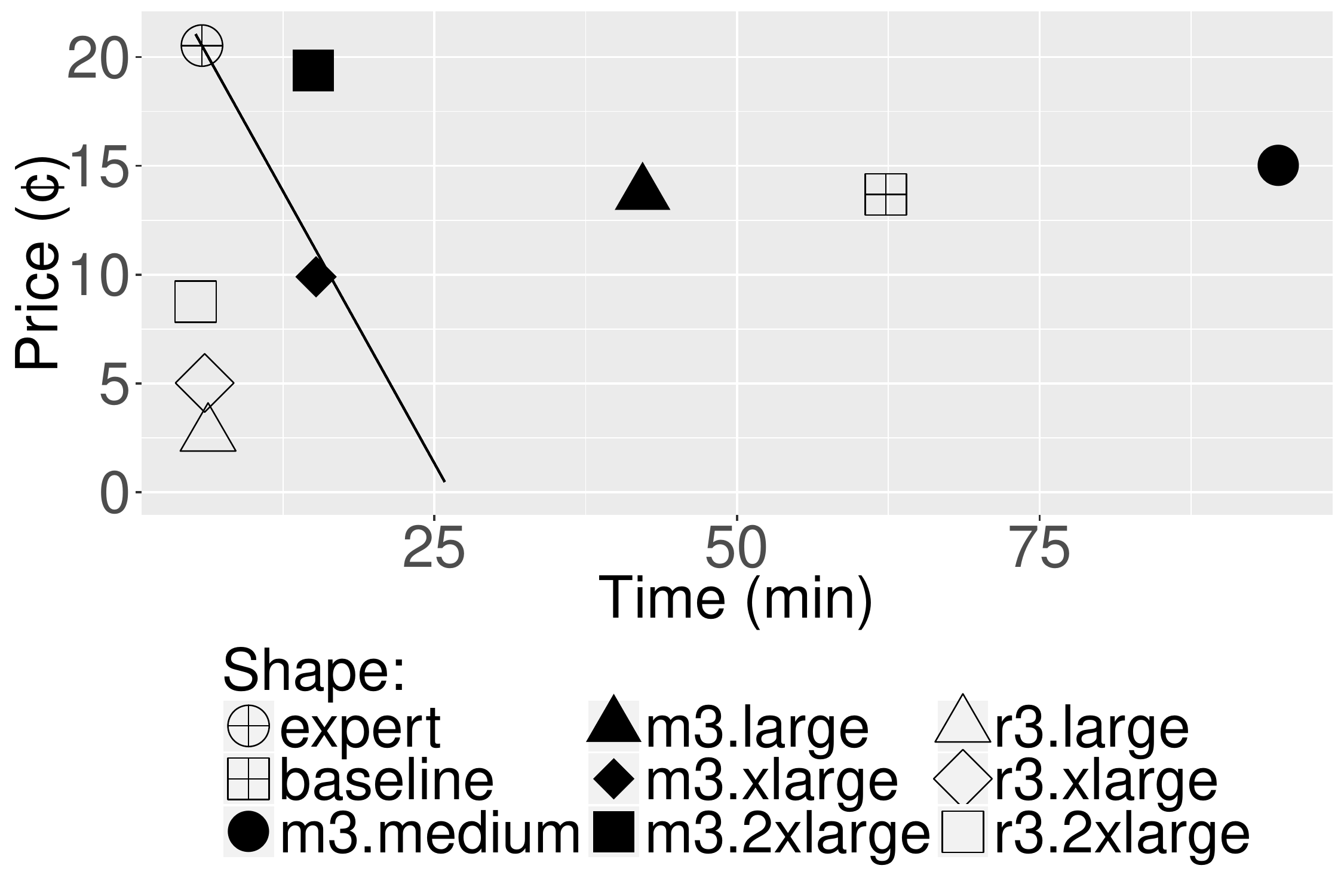}
                \caption{Utility trade-off (RDS)}
                \label{fig:usabilityRDS_UScatter}
        \end{subfigure}
        \begin{subfigure}[b]{0.66\textwidth}
                \includegraphics[width=\textwidth]{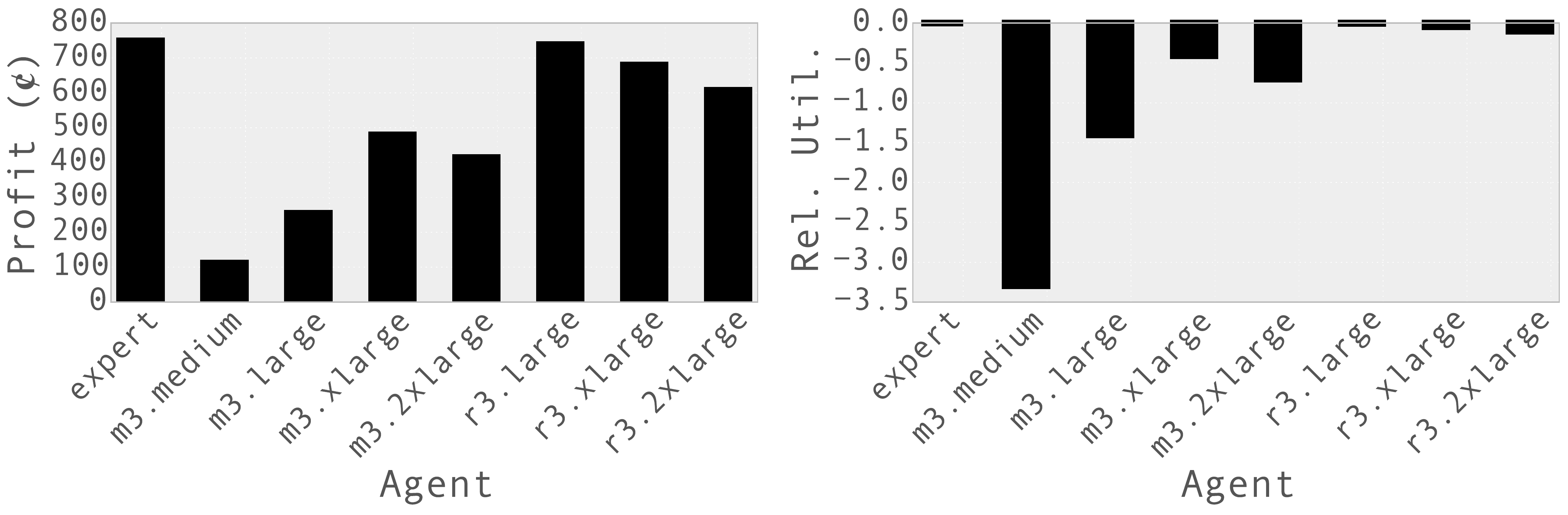}
                \caption{Agent profits and relative utilities (RDS)}
                \label{fig:compRDS}
        \end{subfigure}
        \begin{subfigure}[b]{0.33\textwidth}
                \includegraphics[width=\textwidth]{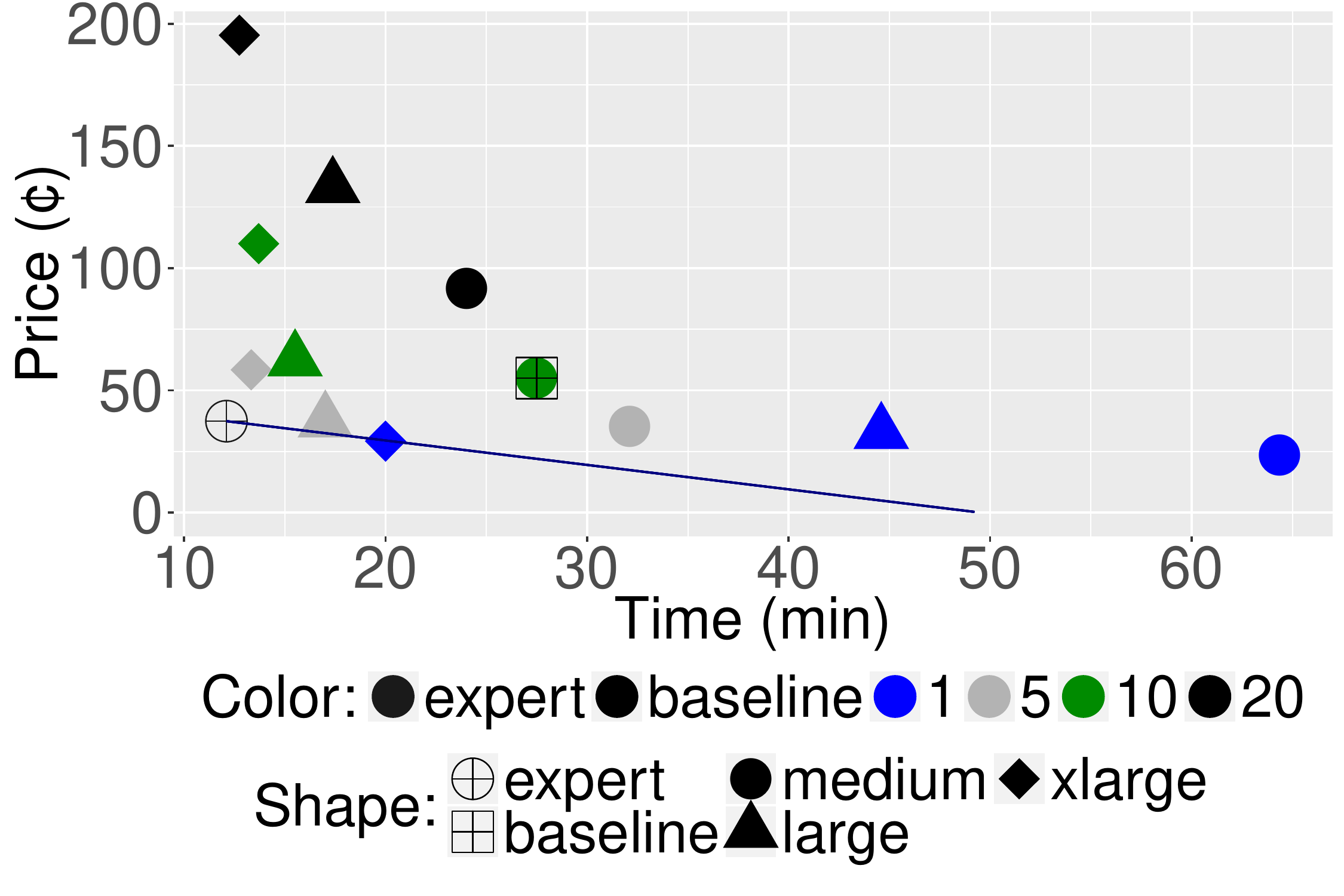}
                \caption{Utility trade-off (EMR)}
                \label{fig:usabilityEMR_UScatter}
        \end{subfigure}
        \hfill
        \begin{subfigure}[b]{0.65\textwidth}
                \includegraphics[width=\textwidth]{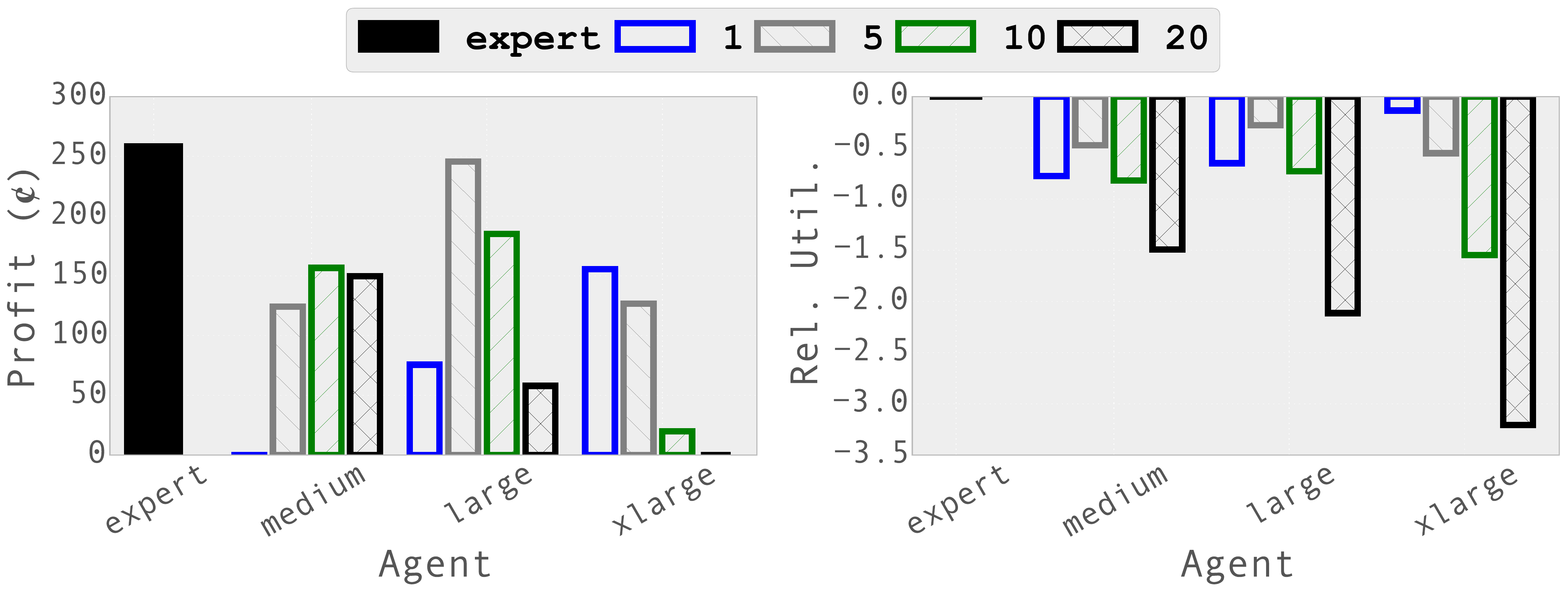}
                \caption{Agent profits and relative utilities (EMR)}
                \label{fig:compEMR}
        \end{subfigure}
        \caption{
        (a, c) Users achieve better utility by using an expert agent, compared to na\"ively selecting a default configuration.  The agent benefits 40\% of the consumers in RDS workloads, and 100\% of the consumers in EMR workloads.
        (b, d) Expert agents always achieve the largest profits.  This means that our market framework gives incentives to agents to find optimal configurations.
        }
\end{figure*}

\subsection{Agent incentives and market properties} 

\looseness -1
In this section we demonstrate that the pricing framework satisfies three
important properties: \emph{competitiveness}, \emph{fairness}, and
\emph{resilience}. These properties incentivize consumers and agents to use
and trust the market by ensuring that (a) the agents will identify efficient
computation plans and provide accurate pricing, and (b) inaccurate estimates
will not pose a great risk to the agents.  A theoretical analysis of the market properties is included in Appendix \ref{apx:linearCase}.

\paragraph{Competitiveness}\smallskip
We run experiments on Amazon RDS and EMR to demonstrate how different
configurations impact profitability in practice. Our goal is to show that, in
our market, well-informed, expert agents can make more profit than na\"ive
agents, thus creating incentives for agents to be competitive and offer
configurations and contracts that benefit the consumers.
In this experiment, a na\"ive
agent selects one configuration to use for all queries. In contrast, the expert agent always selects the optimal configuration for
each query. 
The goal of this experiment is to show the impact of configuration selection.
Thus we control for other parameters, such as the accuracy of the agents'
estimates. So, for now, we assume that all agents know the distributions of
time and cost accurately. We relax this assumption in later experiments.

\begin{figure}[t]
\begin{center}
\begin{tabular}{l l  l }
\toprule
& \textbf{Task} & \textbf{Configuration} \\
\midrule
\multirow{3}{*}{\textbf{RDS}}
& q1--q3, q5--q16, q18, q22  & db.r3.Large \\
& q4, q17, q19, q20 & db.r3.xLarge \\
& q21 & db.r3.2xLarge \\
\midrule
\multirow{3}{*}{\textbf{EMR}}
& WordCount & m1.Medium $\times$ 20 \\
& Sort & m1.Large $\times$ 5  \\
& Join & m1.xLarge $\times$ 1 \\
\bottomrule
\end{tabular}
\end{center}
\vspace{-2mm}
\caption{The expert agents select different configurations for different tasks to maximize profit.}
\label{tbl:bestConf}
\end{figure}%

We generate histograms of time and cost by evaluating each query with each
configuration 100 times. All agents use these histograms to approximate the
distributions and price contracts based on these distributions.
After an agent prices a contract, 
we compute the number of accepted contracts according to the demand function, $\Demand(\Util)$.  Then we randomly select $\Demand$ executions to do trials. The agent receives payments based on whether the execution met or missed the deadline. 

Figure~\ref{fig:compRDS} illustrates the total profit made by each agent
pricing RDS workloads. There are 7 na\"ive agents, each using one of the RDS
configurations from Figure~\ref{tbl:amazonMachines}, and one expert agent, who
always uses the best configuration for each task. Figure~\ref{fig:compEMR}
illustrates the same experiment on EMR workloads. We use one expert agent and
12 na\"ive agents who used the three EMR configurations from
Figure~\ref{tbl:amazonMachines}, each with a cluster size of 1, 5, 10, or 20
nodes. Figure~\ref{tbl:bestConf} lists the configuration chosen by the expert agent for each RDS and each EMR task. In both experiments, the expert agent achieves the \emph{highest} overall
profit.

Figures~\ref{fig:compRDS} and~\ref{fig:compEMR} also show the utilities
offered by the agents for the same contracts. We plot the relative utility of
each na\"ive agent, using the utility of the expert agent as a baseline:
$\frac{\mathit{AgentUtility} -
\mathit{ExpertUtility}}{|\mathit{ExpertUtility}|}$. On both RDS and EMR workloads, the
utility offered by the expert was the best among all agents.

Our experiments on both RDS and EMR demonstrate that expert agents achieve
better utility and profit than all other agents. This verifies empirically
that our market design ensures incentives for agents to improve their
estimation techniques and configuration selection mechanisms. This benefits
both consumers, who get better utility, and agents, who get more profit.

\begin{figure}[t]
        \centering
        \begin{subfigure}[b]{0.315\textwidth}
                \includegraphics[width=1.\columnwidth]{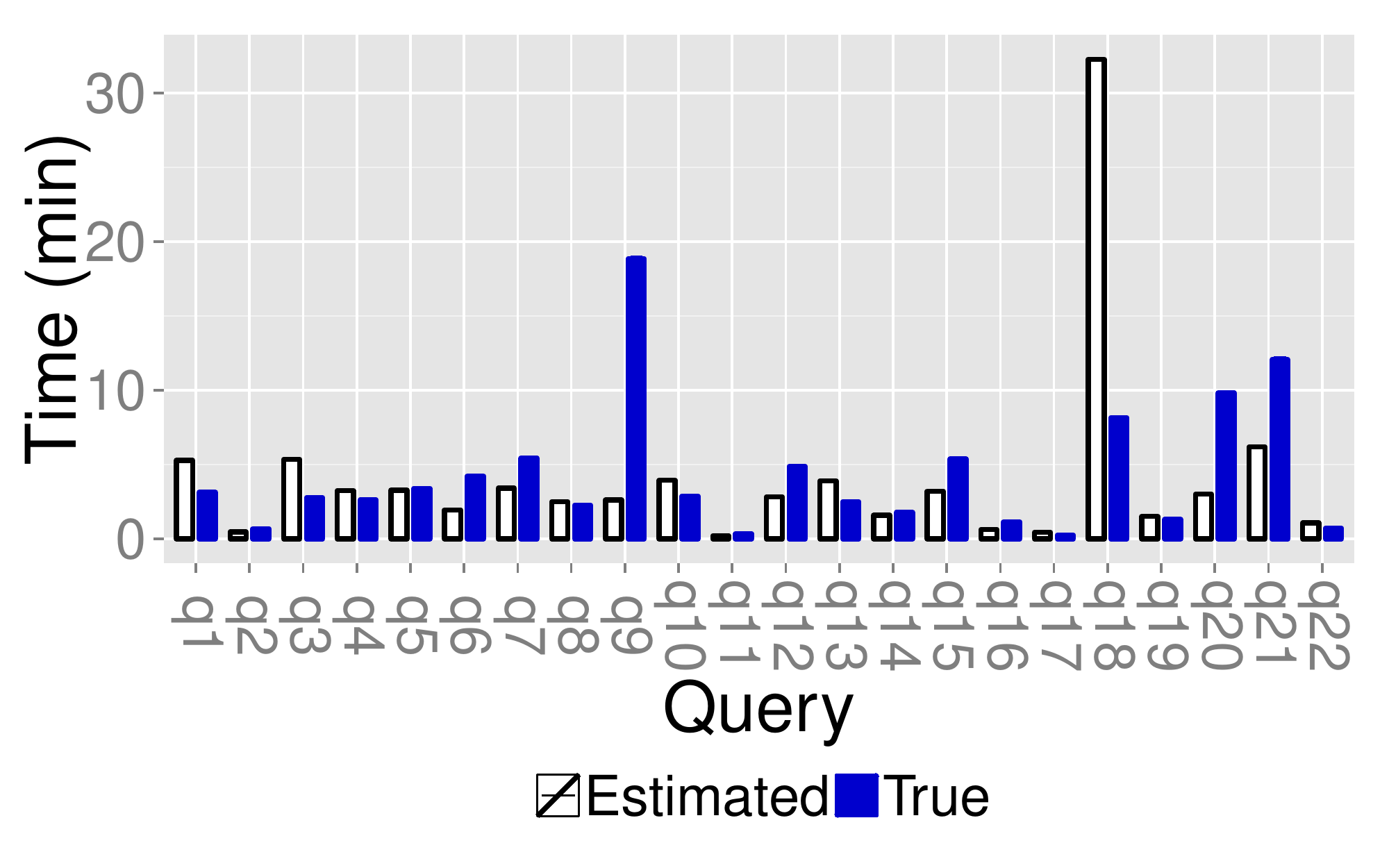}
                \caption{Estimate (RDS)}
                \label{fig:estiRDS}
        \end{subfigure}
        \begin{subfigure}[b]{0.155\textwidth}
                \includegraphics[width=1.\columnwidth]{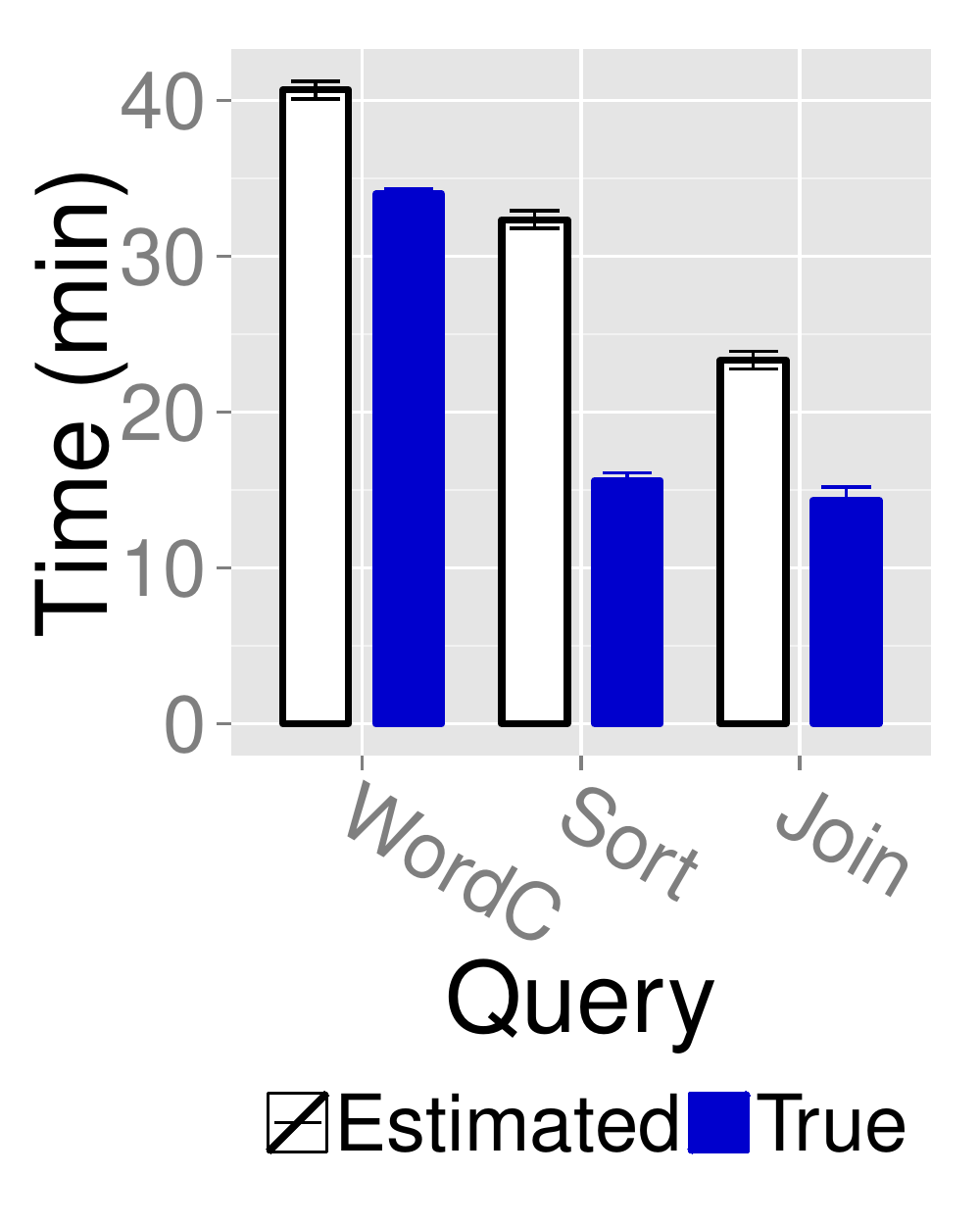}
                \caption{Estimate (EMR)}
                \label{fig:estiEMR}
        \end{subfigure}
        
        \begin{subfigure}[b]{0.23\textwidth}
                \includegraphics[width=1.\columnwidth]{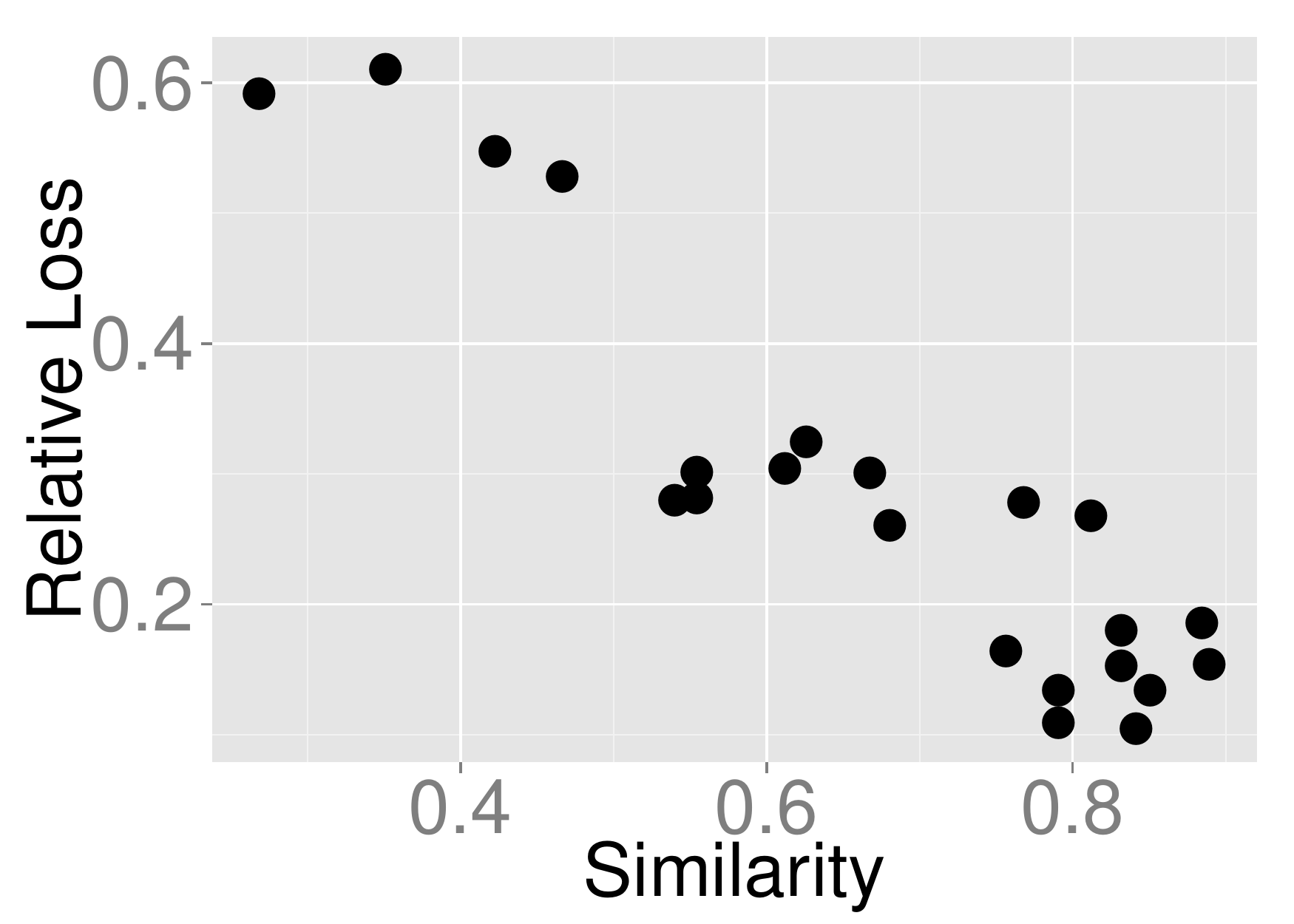}
                \caption{Fairness (RDS)}
                \label{fig:fairRDS}
        \end{subfigure}
        \begin{subfigure}[b]{0.23\textwidth}
                \includegraphics[width=1.\columnwidth]{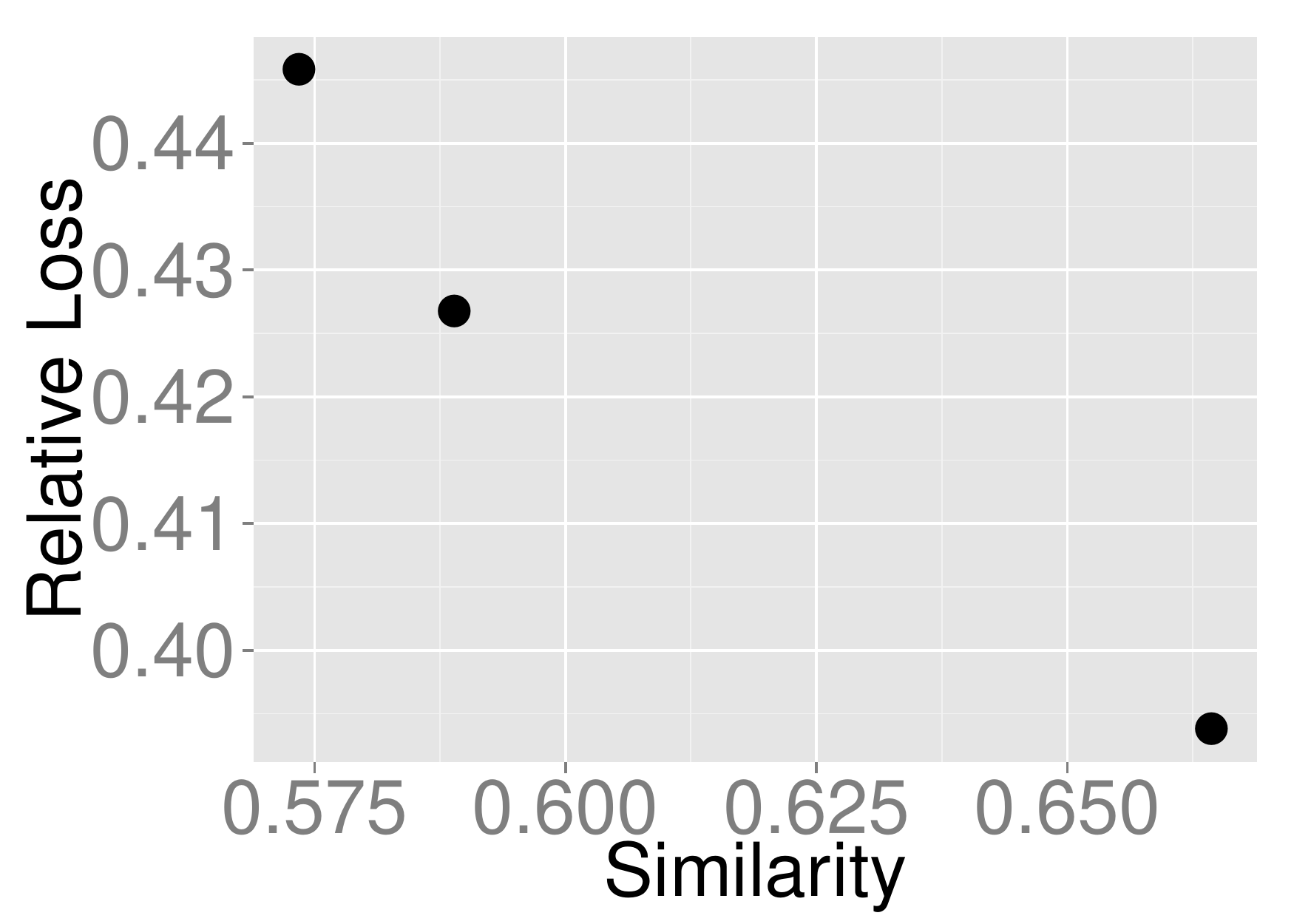}
                \caption{Fairness (EMR)}
                \label{fig:fairEMR}
        \end{subfigure}
        \caption{Agents' estimates are often inaccurate, and such inaccuracies can lead to loss of profit.}
        \label{fig:fairness}
\end{figure}

\paragraph{Fairness}\smallskip 
Fairness guarantees the incentive for agents to present accurate estimates to consumers.
If the agent uses inaccurate estimates, she/he will be penalized with lower profits. Our goal is to show that more accurate estimates lead to greater profit for the agent in practice. 

We consider an agent using db.m3.medium on RDS and PostgreSQL's default query
optimizer to estimate the completion times of queries.
The PostgreSQL optimizer provides an estimate of the expected completion time
and the agent assumes a Gaussian distribution with a mean value equal to the
completion time predicted by the optimizer. We chose $0.05$ for the standard
deviation, which is very close to the actual average standard deviation of the
distributions of the 22 TPC-H queries ($0.04$).

We also consider another agent using m1.medium on EMR, with one master and one slave node. 
The agent estimates the expected completion time by executing queries on a 5\% sample of the data, and assumes a Gaussian distribution around the estimated mean.
The agent uses an empirical standard deviation, $0.55$, which is close to the average true standard deviation of all three EMR job types ($0.56$).

We compare the agents' estimate with the true distributions in
Figures~\ref{fig:estiRDS} and~\ref{fig:estiEMR}. We plot the average
completion time for each TPC-H query and each EMR task. The standard deviation
is very low (under 0.75 min) for all tasks.
 As these plots show, the
agents' estimates can often be far from the actual completion times (e.g.,
$q_{18}$).

Next, we use the similarity between two distributions and relative loss to show the relationship between estimation accuracy and profit.  We compare the true distribution of completion time (which is a histogram) with the agent's estimate (a Gaussian distribution) by turning the agent's estimate into a histogram and computing the cosine similarity between two histograms. The relative loss measures how much profit the agents lose compared to the optimal profit they could have made. We define relative loss of profit as:
\begin{equation}
\label{eq:relativeLoss}
Relative Loss = \frac{Optimal Profit - Actual Profit}{Optimal Profit}
\end{equation}
As Figures~\ref{fig:fairRDS} and~\ref{fig:fairEMR} illustrate, when the agent's estimate is more accurate, the relative loss is smaller.

\begin{figure}
 \centering
 \begin{subfigure}[b]{.46\columnwidth}
     \includegraphics[width=\textwidth]{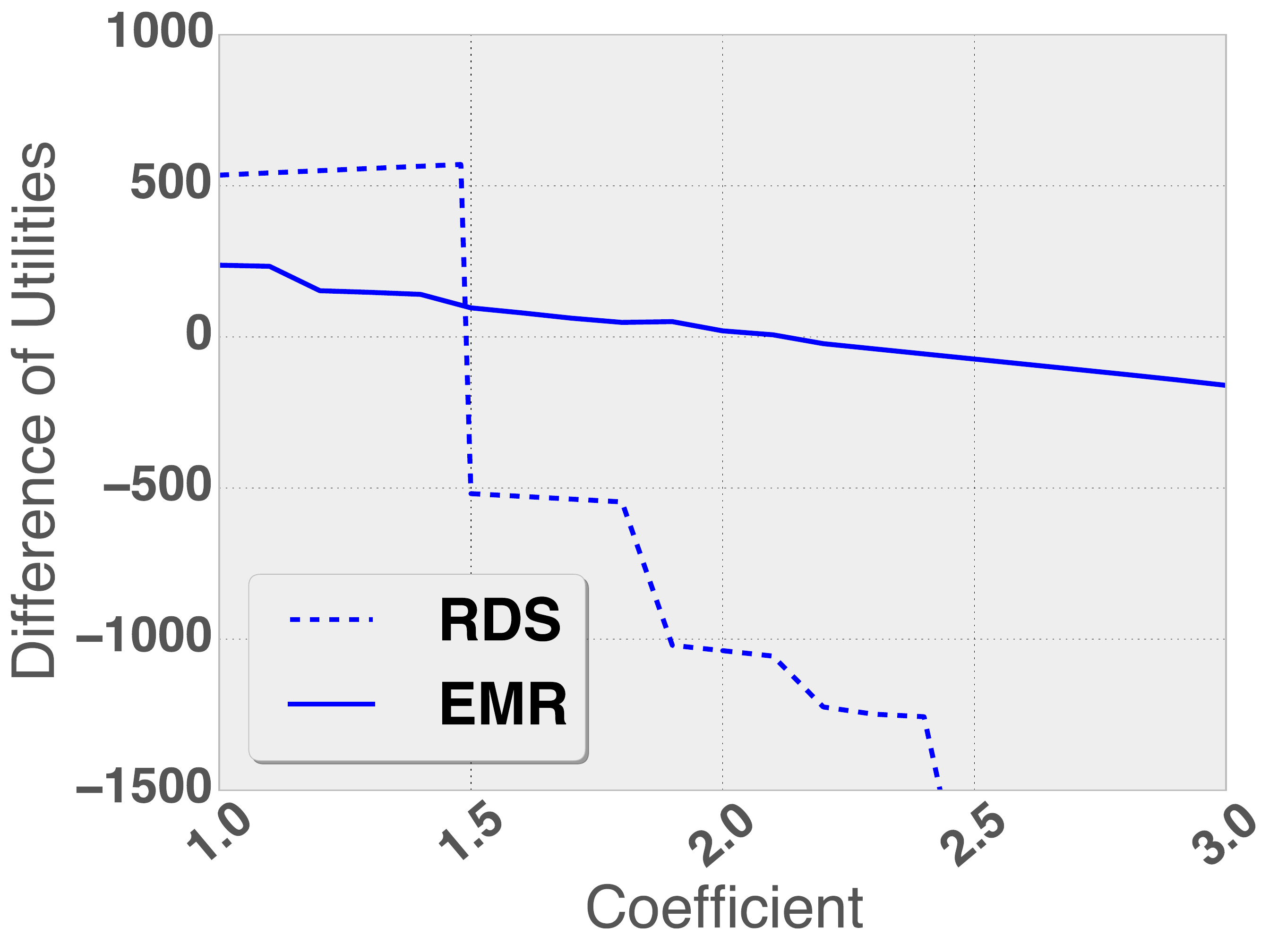}
     \caption{(Agent's utility - consumer's utility) drops in overestimation.}
     \label{fig:overestimate}
 \end{subfigure}
 \begin{subfigure}[b]{.46\columnwidth}
     \includegraphics[width=\textwidth]{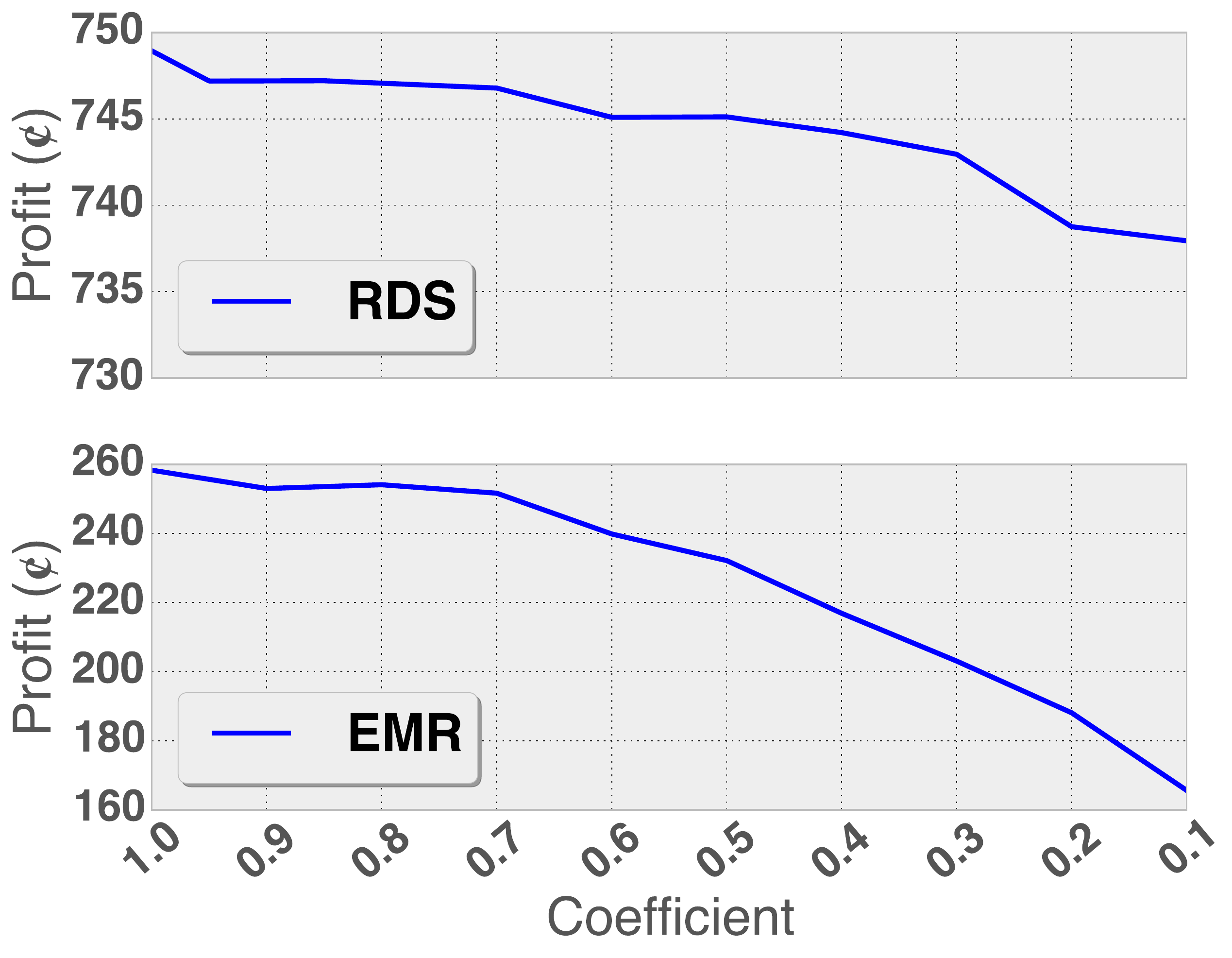}
     \caption{Agent's profit decreases moderately in underestimation.}
     \label{fig:underestimate}
 \end{subfigure}
 \caption{Poor estimation moderately impacts the market.}
 \end{figure}

Our market does not rely on the assumption that the estimates are
accurate, and it can in fact tolerate inaccuracies well. As long as
there exists at least one task for which an agent can produce better
estimates than a consumer, the agent offers utility to the market. In
our experimental evaluation, we showed that this is easy to achieve in
practice: even agents using simple estimation methods (such as using
the PostgreSQL optimizer or sampling), which result in fairly
inaccurate estimates, can provide benefit to non-expert consumers.
Existing research has shown that time and cost estimation is
non-trivial~\cite{Aboulnaga:2009, Herodotou:2011:SOCC, Wu:2014},
and agents using such specialized tools would always produce better
estimates than non-expert consumers.

In addition, we further expanded our evaluation to study an extreme
case: when all agents in the market make worse estimation than all
consumers, for all tasks. We multiply the agents' estimated time and cost by a coefficient.
A $>1$ coefficient means overestimation and a $<1$ coefficient means underestimation.

As depicted in Figure~\ref{fig:overestimate}, overestimation leads agents to post higher
prices lowering the consumers' utilities. However, switching to using
the cloud provider directly becomes preferable (on average) only when
agents overestimate substantially: in our empirical simulation, agents
had to overestimate time and cost by 49\% in RDS workloads, and by
120\% in EMR workloads before a switch was beneficial to consumers on
average.

On the other hand, figure~\ref{fig:underestimate} shows that underestimation of time and cost decreases an agent's
profit by 2\% in RDS and 36\% in EMR if it underestimates time and cost by a factor of 10.
Depending on the agents' profit margins, they may be able to
absorb the difference without losses. To avoid losses, agents can
follow risk-aware pricing strategies (Section~\ref{sec:riskAware}).

\begin{figure}
        \centering
        \begin{subfigure}[b]{0.23\textwidth}
                \includegraphics[width=1.\columnwidth]{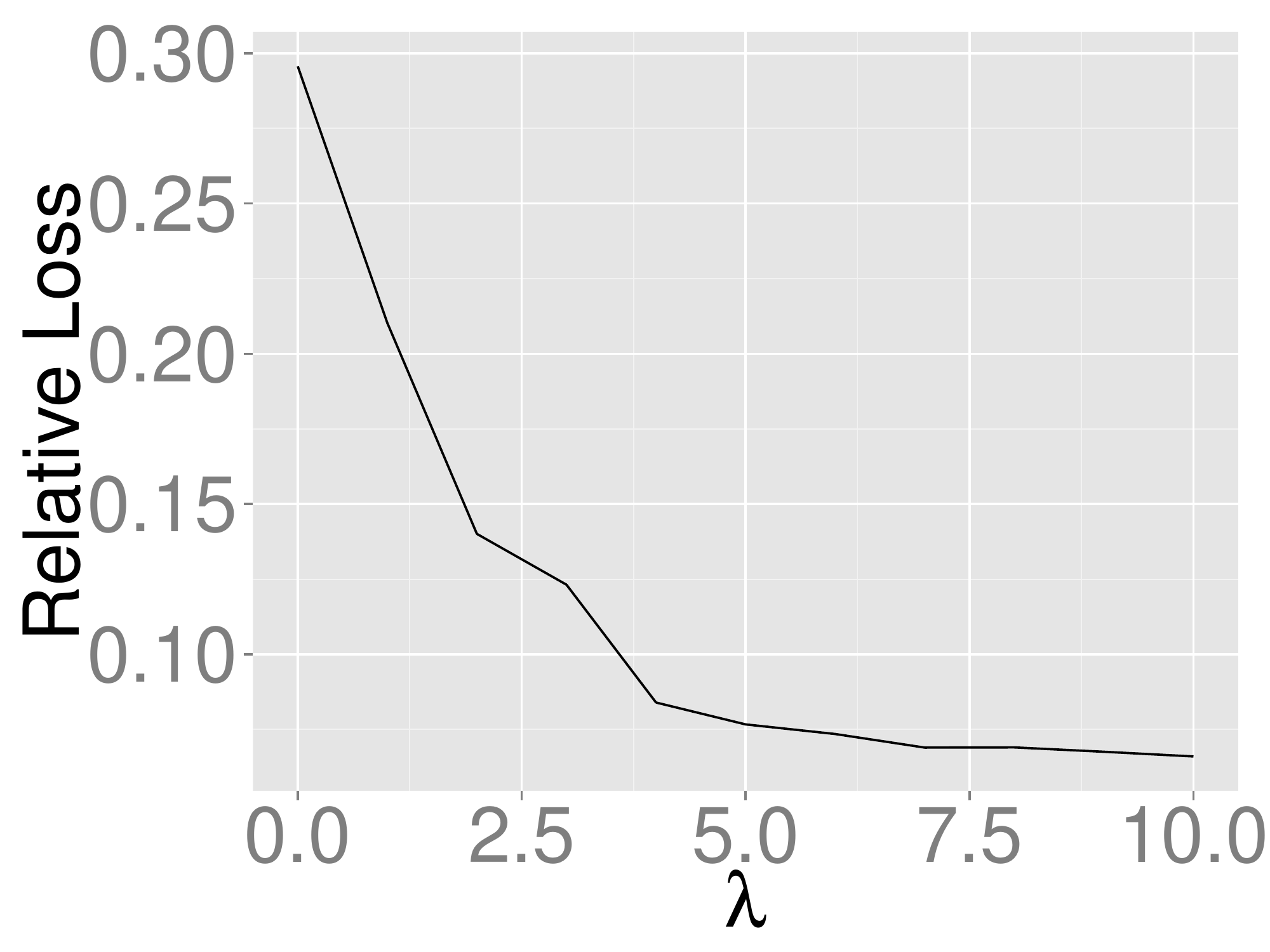}
                \caption{Resilience (RDS)}
                \label{fig:resiRDS}
        \end{subfigure}
	\begin{subfigure}[b]{0.23\textwidth}
                \includegraphics[width=1.\columnwidth]{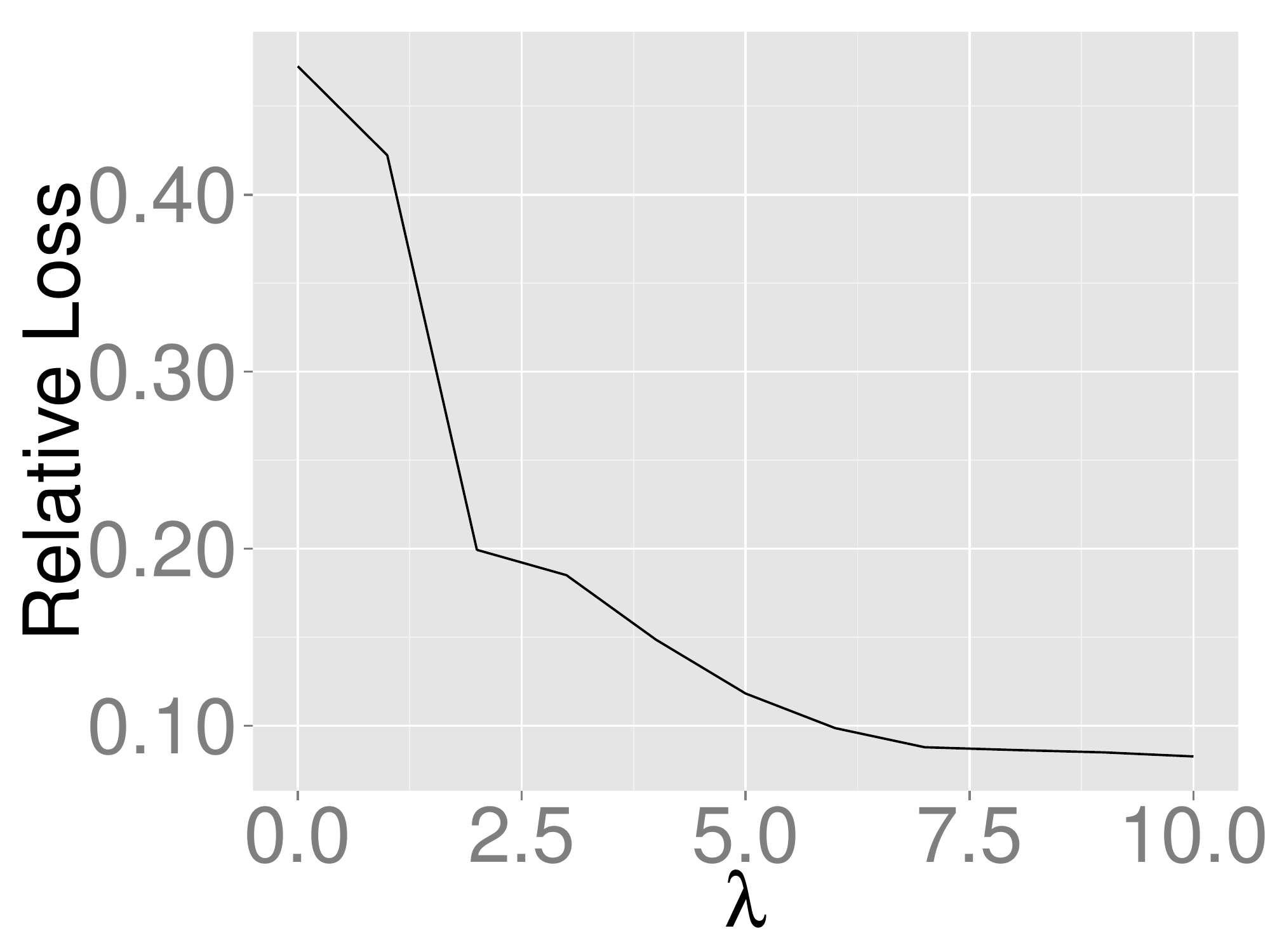}
                \caption{Resilience (EMR)}
                \label{fig:resiEMR}
        \end{subfigure}
        \caption{By adjusting for risk, agents can reduce their losses in case of inaccurate estimates.}
        \label{fig:resilience}
\end{figure}

\begin{figure}
        \centering
        \begin{subfigure}[b]{0.23\textwidth}
                \includegraphics[width=1.\columnwidth]{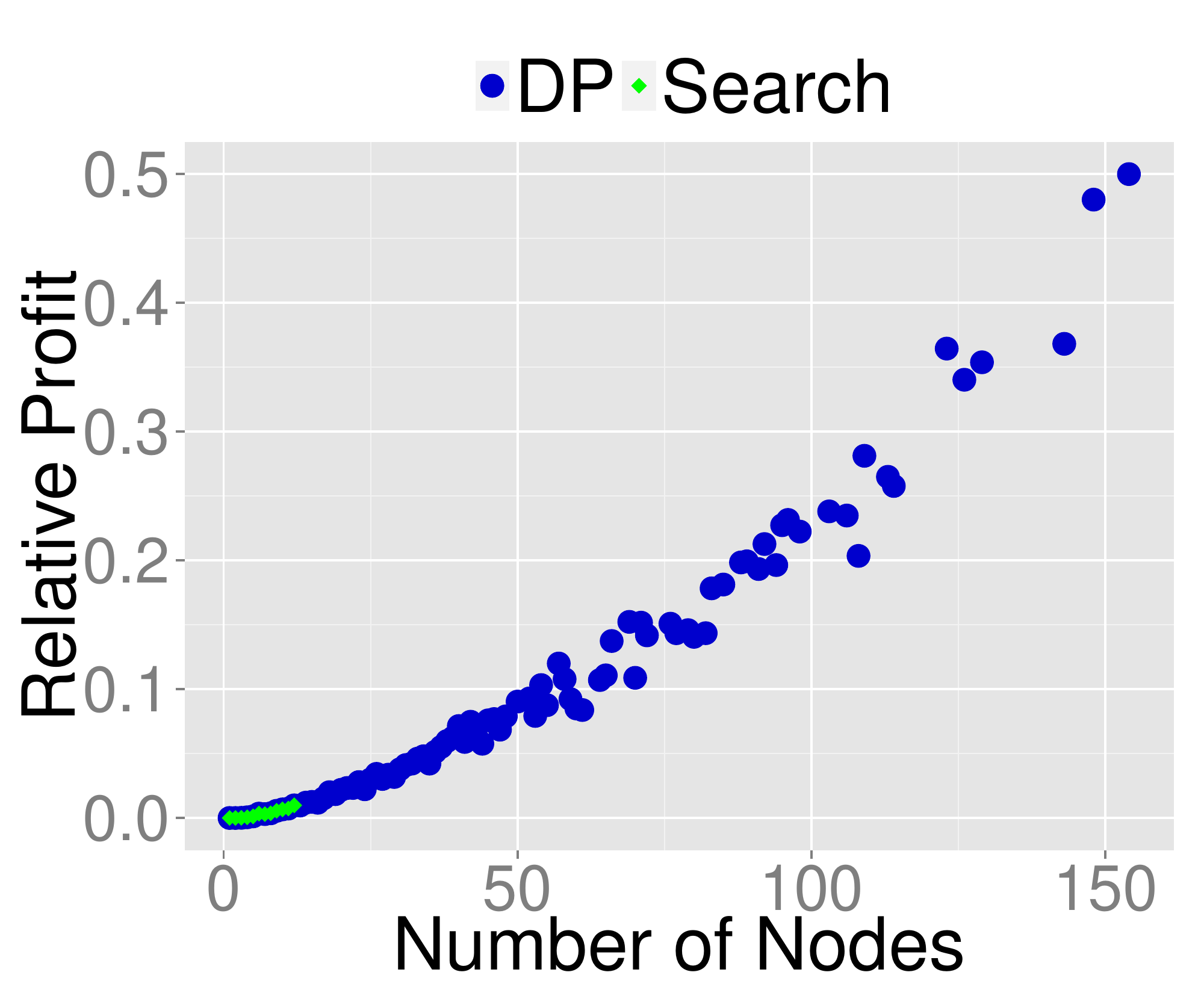}
                \caption{Profit Relative to Greedy}
                \label{fig:dpProfit}
        \end{subfigure}
	\begin{subfigure}[b]{0.23\textwidth}
                \includegraphics[width=1.\columnwidth]{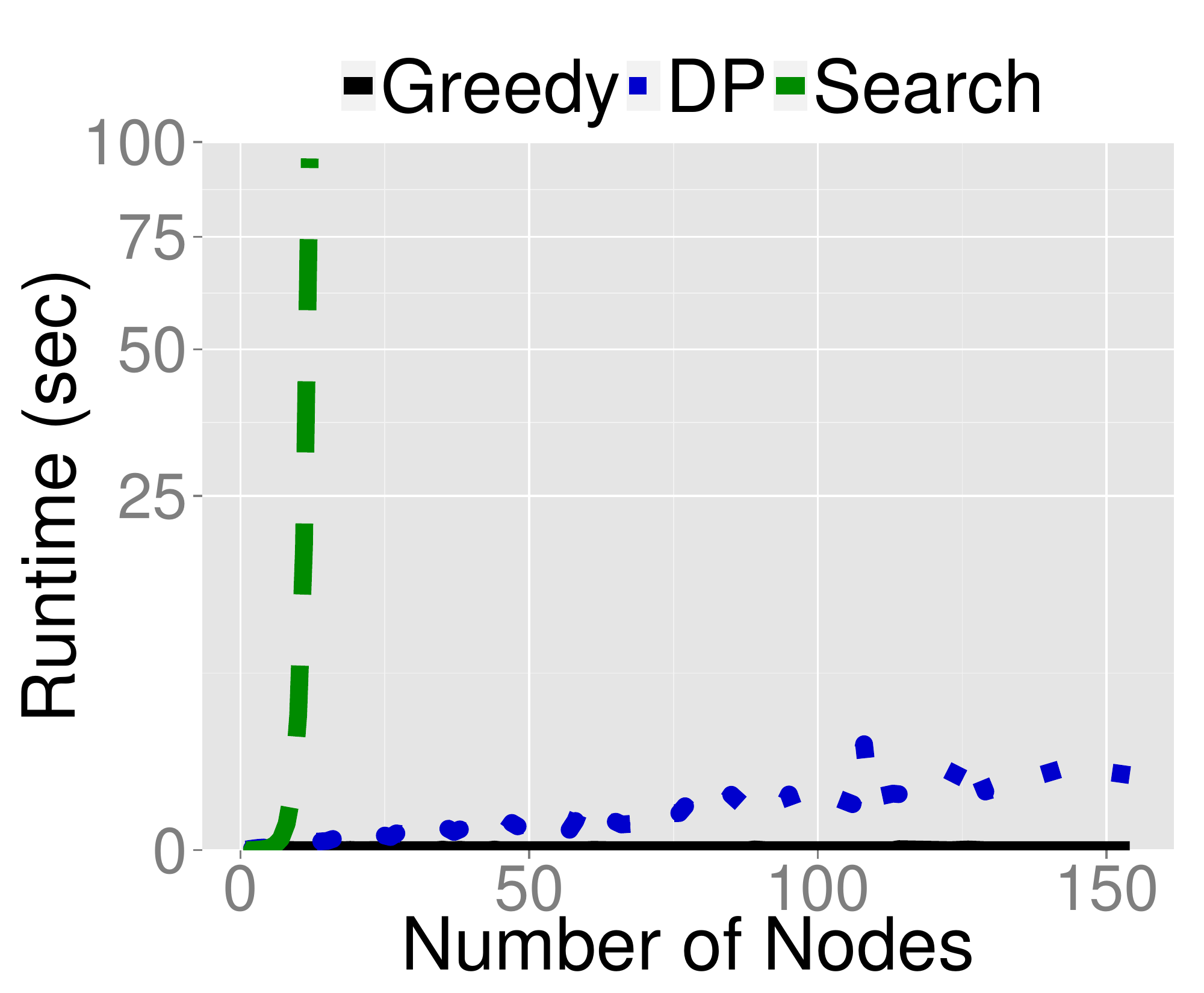}
                \caption{Runtime}
                \label{fig:dpTime}
        \end{subfigure}
        \caption{DP outperforms Greedy and Search}
        \label{fig:dynamicProgramming}
\end{figure}

\paragraph{Resilience}\smallskip
The property of \emph{resilience} provides assurances to the agents, by
ensuring that inaccurate estimates will not pose a significant risk to the
agents' profits.
This property is crucial, as errors in the estimates are
very common~\cite{Akdere:2012, Duggan:2013, Wu:2014}. Our framework ensures
resilience to these inaccuracies by accounting for \emph{risk} (Definition~\ref{def:risk}).
Specifically, the agents can profit by adjusting the risk they
prefer to take. According to Equation~\ref{eq:riskAware}, the risk is part of
the objective and controlled by a parameter $\lambda$. When $\lambda$ is
large, the agent has low confidence in the estimate (conservative). This
setting reduces the loss of profit if the agent's estimate is inaccurate.

We again consider an RDS agent
using db.m3.medium and the default PostgreSQL optimizer, and an EMR agent
using m1.medium and sampling to estimate runtime. We evaluate relative loss
using Equation~\ref{eq:relativeLoss} and plot it for different values of
$\lambda$ (Figure~\ref{fig:resilience}). A value of $\lambda=0$ means that the
agent is confident that their estimate is correct. However, since in this case
the estimates were inaccurate, the relative loss for $\lambda=0$ is high: the
agents' profit is much lower than the optimal profit they could have achieved.
For both agents (EMR and RDS), the relative loss decreases for higher values
of $\lambda$.  
This shows that by adjusting for risk, the agents can reduce
loss of profit.

\subsection{Fine-grained pricing} \label{sec:fineExperiments}
In our final set of experiments, we evaluate fine-grained pricing
(Algorithm~\ref{alg:fineGrained}) against a large dataset of real-world
scientific workflows \cite{DeRoure:2009}. This dataset is well-suited for this experiment, as it
provides diverse computational flows of varied sizes and complexities. The
published workflows do not report real execution information (time and cost),
and we are not aware of any public workflow repositories that provide this
information. Therefore, we augment the real workflow graphs with synthetic time and cost
histograms for each subtask, drawn from random Gaussian distributions with
means in the [1,100] range, and variances in the [0,5] range. 
Each subtask has 5 candidate configurations with different time and cost
histograms.
We compute the profit using utility $\Util(t, \Price) = - t - \Price$ and
demand $\Demand(\Util) = 100 + 0.01 \Util$ (Section
\ref{sec:linearCaseSolution}). We set $\lambda_\Demand$ (the coefficient of $\Util$) to a smaller value than the ones used for RDS and EMR workloads, because the completion times and costs for workflows are much larger.

\begin{figure}
	\centering
        \includegraphics[width=.75\columnwidth]{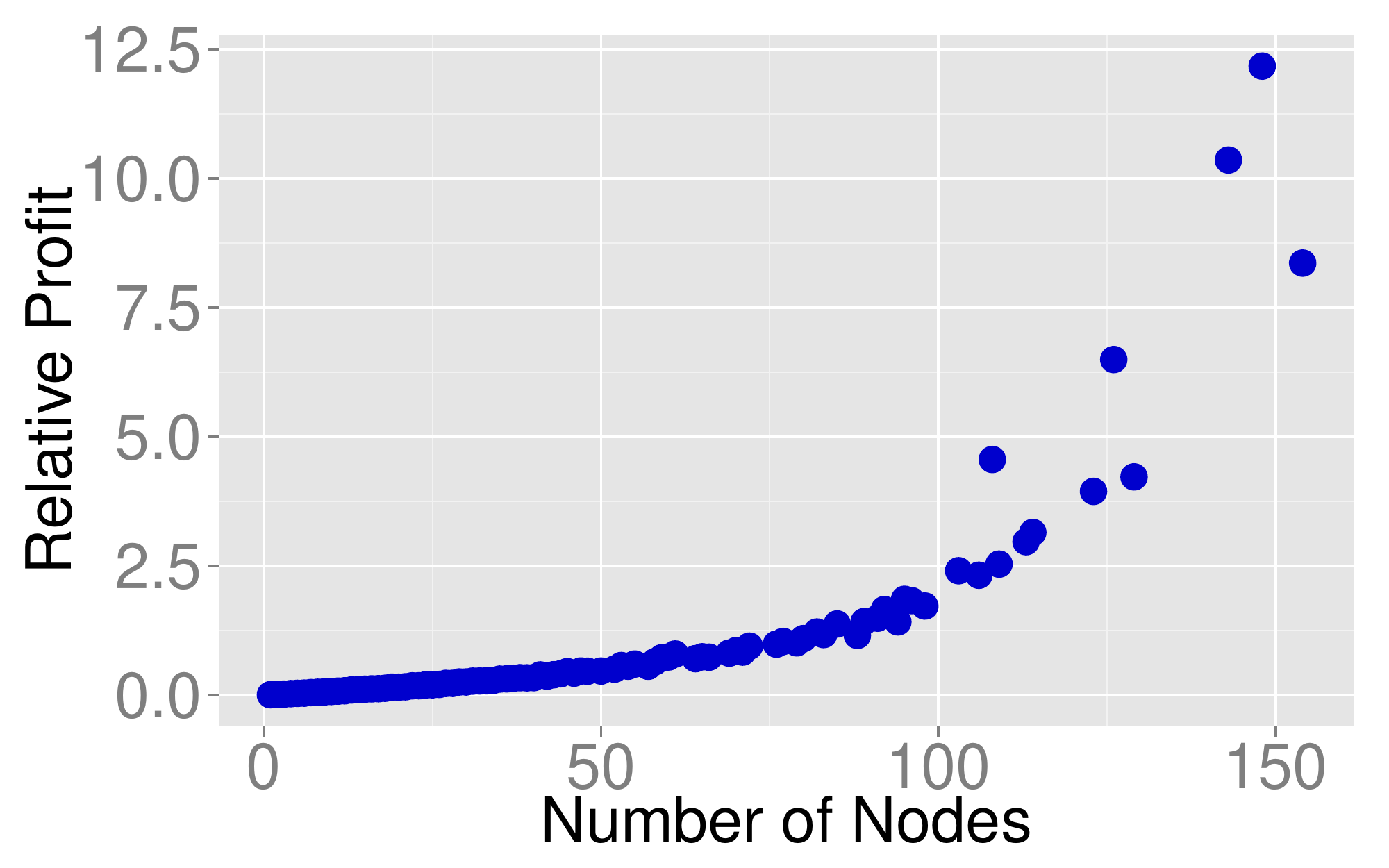}
        \caption{Pricing at finer granularities can vastly increase the agents' profits.} 
        \label{fig:coarse}
\end{figure}

First, we evaluate our Dynamic Programming algorithm
(Algorithm~\ref{alg:fineGrained}) against two baselines: (1) an exhaustive
search strategy (\emph{Search}) that explores all possible configuration
assignments, and (2) a greedy strategy (\emph{Greedy}) that selects the
configuration that leads to the maximum local profit for each subtask. We
perform 10 repetitions for each workflow, using different random time and cost
distributions for each repetition.
Figure~\ref{fig:dpProfit} shows the relative profit achieved by \emph{Search}
and \emph{DP} compared to \emph{Greedy}: $\frac{\ProfitAll -
\ProfitAll_{Greedy}}{\ProfitAll_{Greedy}}$. \emph{DP} achieves better profits
than \emph{Greedy}, and the effect increases for larger workflows: for
workflows with 154 subtasks, \emph{DP} achieves 50.0\% higher
profit than \emph{Greedy}. \emph{Search} provides few data points, as it
cannot scale to larger graphs. For small workflows (up to 12 subtasks)
\emph{Search} and \emph{DP} select equivalent configurations that result in
the same (optimal) profit. Figure~\ref{fig:dpTime} shows the running time of
the three algorithms. As expected, exhaustive search quickly becomes
infeasible, and \emph{Greedy} is faster than \emph{DP}. However, the runtime
of \emph{DP} remains low even for larger workflows. Combined with the profit
gains over \emph{Greedy}, this experiment demonstrates that
Algorithm~\ref{alg:fineGrained} is highly effective for fine-grained pricing.

Second, we evaluate the benefits of fine-grained pricing, compared to
coarse-grained pricing. Figure~\ref{fig:coarse} shows the profit achieved by
\emph{DP}, which assigns a configuration to each subtask, relative to the
optimal single configuration for the entire workflow. In this experiment,
fine-grained pricing doubled the agents' profits for small workflows,
compared to coarse-grained pricing, and the gains increase as workflows grow
larger. For the largest workflows in our dataset, fine-grained pricing
achieved $12.5$ times higher profits.

\subsection{Comparison with Alternative Approaches} \label{sec:extraEval}

\begin{figure}
\centering
    \includegraphics[width=.6\columnwidth]{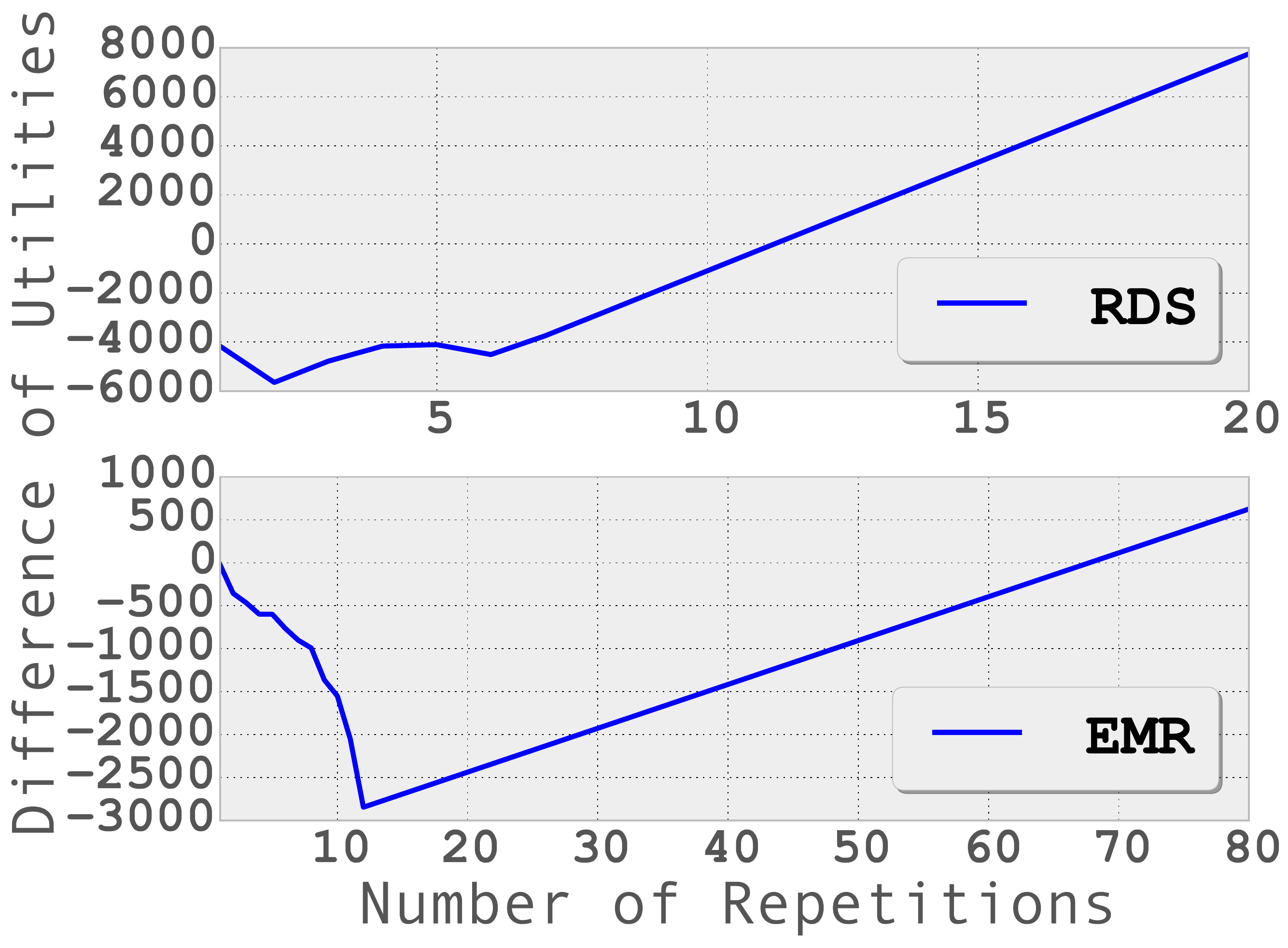}
    \caption{Consumers prefer our approach to benchmark except for highly repetitive workloads}
    \label{fig:benchmark}
\end{figure}

\paragraph*{Benchmark-based Approach} ~

Contrasting our work with Benchmarking as a Service (BaaS)~\cite{Floratou:2011} is
meaningful when workload repetition is significant. We
assume a consumer who repeate RDS and EMR workloads without any
modifications, and with each repetition tested a different
configuration; once all configurations were tested, the consumer would
continue using the best one in subsequent repetitions. For this
experiment, we limited the number of possible configurations to
7 for RDS and 
12 for EMR. This biases the experiment in
favor of benchmarking, as in practice the number of configurations
that the consumer would have to try is much higher. In this simplified
setting, we found that it took 12 repetitions in RDS and 68
repetitions in EMR before the consumer would start benefiting from
benchmarking. In the real world, these numbers are much higher, as
cloud providers offer way more configurations than the ones we
considered here. Cluster size alone causes an explosion in the number
of options, so having an agent with an analytical model, such as in
\cite{Herodotou:2011:SOCC}
, is necessary.

In practice, BaaS has additional challenges: As discussed in~\cite{Floratou:2011}, data growth and changes in the input make
BaaS complicated. Workloads are almost never repeated exactly, as the
input changes between executions, requiring the BaaS provider to
monitor and react to changes. Moreover, cloud providers change machine
types, parameters, and pricing very frequently\,---\,e.g., between 2012 and
2015, AWS introduced on average 2.6 new instances every three months.
When these settings change, resource selection needs to be
re-evaluated, even if a workload stays the same.

\paragraph*{VCG-auction-based Approach} ~

In this experiment, we compare our model with an VCG auction model. In a strategy-proof VCG auction, the agents truthfully reveal their best costs of executing a task. Then the consumer selects the agent with the best utility but pays according to the second best utility. Therefore, given a specific task, only the best and the second best agents' contract together define the price. So we create one agent who is always able to find the best configuration for each task, and another agent who is doing just worse than the best one. We assume a delta $\Delta$ that is the difference between the best utility and the second best utility, and vary this delta to see how much profit the best agent can get.

\begin{figure}
\includegraphics[width=\columnwidth]{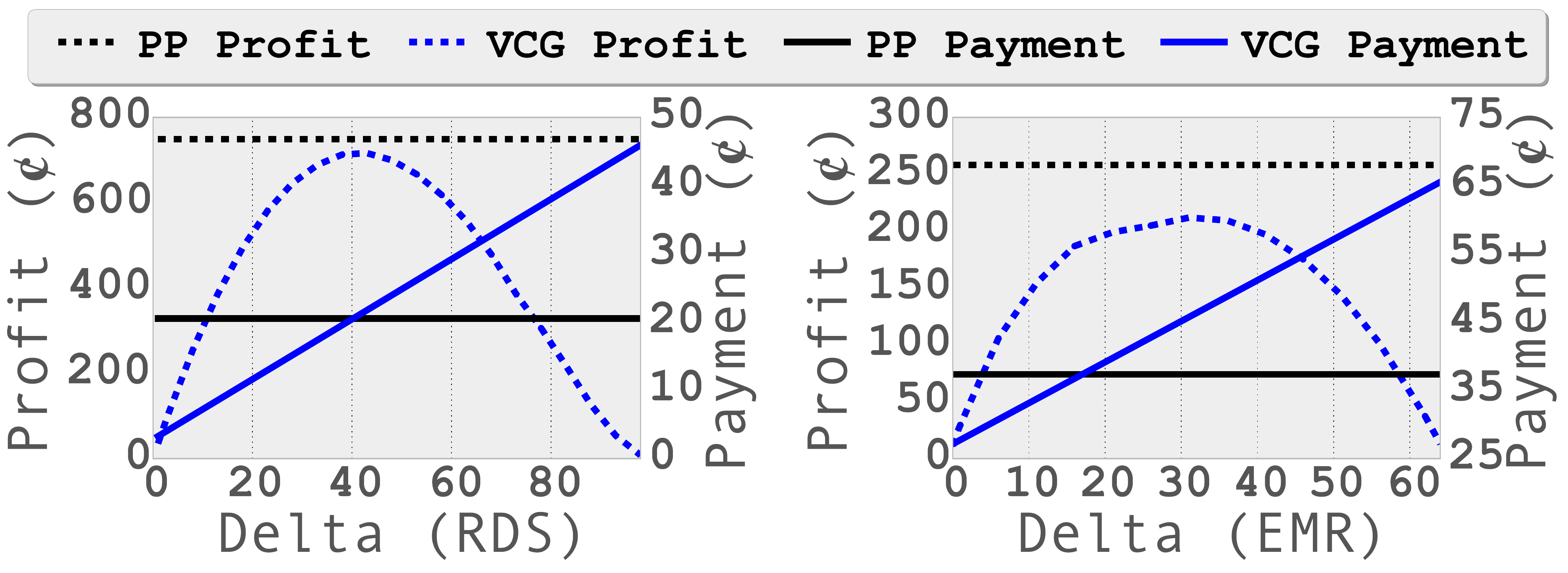}
\caption{VCG auction brings less profit to agents without necessarily reducing consumers' payments. (PP = our posted-price approach)}
\label{fig:vcg}
\end{figure}

As depicted in Figure \ref{fig:vcg}, when $\Delta$ becomes greater, 
the best agent's profit goes up but then drops down. This is because of the demand. Larger $\Delta$ means more profit of each contract but less demand. So the profit reaches maximum at a certain point. Even this maximum VCG profit is less than the profit in our approach. It is because our approach optimizes the profit for each individual task in a workload, while the VCG approach in this experiment applies unified profit. Moreover, when $\Delta$ gets greater, consumers' average payment for each contract is bigger. This is by definition of VCG. The VCG's average payment goes higher than our approach's when $\Delta$ is very small.
In a word, our approach brings more profit to agents without necessarily increasing consumer's payment.

%% file: related.tex

\section{Related Work} \label{sec:related}
In contrast to our market framework, which emphasizes the consumer need for task-level pricing, 
existing work on cloud pricing largely focuses on resource 
usage. One study used game theory to model a pricing
framework where consumers compete with each other to
maximize their utilities~\cite{Daoud:2009, Hadji:2011}. Specifically, each consumer
has a demand on resources, and their utility is a function of
demand and price. A choice of price by a service provider
triggers a change in the consumers' demands to maximize
their utilities, thus affecting the provider's revenue. This
work makes two key assumptions that are not present in our
framework. First, the chosen utility functions indicate that
the quality of service (QoS) degrades when consumers share
resources. While meaningful for resources such as wireless
bandwidth, this assumption has been shown to not always hold in
many types of resources relevant to computation~\cite{Ahmad:2011:VLDB, Ahmad:2011:EDBT}. In
fact, QoS can improve when, for example, consumers share
data and cache, and agents in our framework can take advantage of this to make more profit. Second, it is assumed that the consumers
know each other's demands and strategies, and adjust their
demands accordingly. In contrast, in our framework we consider consumers' tasks separately and use probability distributions to model runtime and financial cost, leading to a
simpler yet practical model.

Variants of pricing mechanisms assume that providers
price dynamically, based on the consumer arrival and departure
rates~\cite{An:2010, Mihailescu:2010:CCGrid, Xu:2012, Xu:2013:CC, Xu:2013:SIGMETRICS}. In turn, prices also guide consumer demand. In a different direction, Ibrahim et al.~\cite{Ibrahim:2011} argue that the interference across virtual machines sharing the same hardware leads to overcharging. They suggest cloud providers to price based on effective virtual machine time. This framework guarantees benefits to consumers and urges providers to improve their system design.

Wong et al.\ \cite{Wong:2012} have compared three different pricing
strategies in terms of fairness and revenue: (1) Bundled
pricing, in which providers sell resource bundles (e.g., virtual machine with CPU, memory, and other resources) to
consumers; (2) Resource pricing, in which providers charge
consumers separate prices for the consumed resources; (3)
Differentiated pricing, in which providers charge consumers
personalized prices. They define fairness as a function of equitability and efficiency of utilities among all consumers and conclude that differentiated pricing provides
the best fairness. They treat consumers' jobs identically
and define fairness based on the number of jobs that are
successfully executed by the cloud provider. They do not
consider the connection between uncertain completion time
and utility.

Economic-based resource allocation has been extensively studied in
grid computing~\cite{Buyya:2005, Buyya:2009, Haque:2011, Qureshi:2014,
Kiefer:2014}. Researchers have developed different economic models in
two main categories: ``commodity markets'' and ``auctions''. In a commodity
market, resources are sold at a posted-price. The price of resources
affects consumers' utility and demand, and therefore impacts the
providers' profit. Finding the equilibrium is the main focus in these
models. Yeo et al. proposed a utility-driven pricing
function~\cite{Yeo:2004, Yeo:2007} and an autonomic pricing
approach~\cite{Yeo:2010}. Stuer et al.~\cite{Stuer:2007} adapted
Smale's method~\cite{Smale:1976} to price resources in grid computing markets.
Bossenbroek et al.~\cite{Bossenbroek:2009} introduced option contracts
into the market and used hedge strategies to reduce consumers' risk of
missing task deadlines. Auction-based pricing in grid computing
contains several different forms. Double auction requires consumers
and providers to publish their requests and offers in a
marketplace~\cite{He:2003, Tan:2007, Wieczorek:2008, Izakian:2010}.
Vickrey auction is a type of sealed-bid auction in which the highest
bidder wins but pays the second-highest bid~\cite{Waldspurger:1992,
Nisan:1998}. Combinatorial auction allows consumers to bid on
combinations of resources~\cite{Chun:2005}.

Posted-price selling and auctions are both
established ways of selling, and it is not clear which one is better.
The key challenge is the uncertainty of the value of the commodity (in
this case, the computational resource) and researchers have developed
different models to compare the two mechanisms under various
assumptions. Computer scientists measure system metrics in these two
mechanisms. Wolski et al.~\cite{Wolski:2001} state that posted-price
brings more price stability, higher task completion rate
and higher resource utilization ratio 
than auctions. Vanmechelen and Broeckhove~\cite{Vanmechelen:2007}
conclude that posted-price results in more stable pricing, while Vickrey
auction results in fewer message passing in dynamic pricing.
Economists have discussed the revenues in these two mechanisms.
Wang~\cite{Wang:1993} compares posted-price selling and auctions where the buyers have independent private values of the commodity. He
finds that posted-price selling brings more profit to the seller when
the buyers' values of the commodity are widely dispersed. Campbell and
Levin~\cite{Campbell:2006} state that auctions perform worse than
posted-price when buyer valuations are interdependent.
Hammond~\cite{Hammond:2010, Hammond:2013} concludes that revenues of
the two mechanisms cannot be statistically distinguished based on his
study on eBay.

In contrast to our framework, this entire body of work focuses on
resource-level pricing, and does not provide a mechanism for consumers to
select resources based on their tasks. Recent work has started shifting the
focus to task-level pricing. 
Floratou et al.~\cite{Floratou:2011} propose a Benchmark as a Service (BaaS) that benchmarks user's workload and suggests the optimal configuration for repetitive execution. As they mention in the paper, changes such as growth of input data make BaaS complicated. In our approach, we do not assume repetitiveness of workloads. Consumers may pay more for extremely repetitive workloads, but are free from benchmarking evolved workloads.

Auction-based models~\cite{Tanaka:2013,
Tanaka:2014} assume that providers bid for service contracts. These
models use the Vickrey-Clarke-Groves auction mechanism, which redefines the payment to the winner and guarantees that all providers report their true cost of providing the service. While this work provides a good model for task-level pricing, it does not consider execution time for tasks. In our framework, we balance the consumers' trade-off of execution time and price through their utility function.

Personalized Service Level Agreements (PSLA)~\cite{Ortiz:2013, Ortiz:2015, Ortiz:2016, Ortiz:2016:arxiv} resemble the
contracts in our framework, and describe a vision of a system that analyzes
consumers' data and suggests to them tiers of service. Each tier describes
three properties: completion time, price per hour, querying capabilities. For
example, a tier on Amazon EMR can be ($<3.5$ minutes, \$0.12/hour, SELECT 1
attribute FROM 1 table WHERE condition). In our framework, consumers do not
subscribe to a tier of service, but rather provide the task they need and the
agent provides a specific price for the task.

When multiple agents find the same best configuration for some tasks, their prices affect each other and finally converge to the Nash Equilibrium in differentiated Bertrand model \cite{rasmusen:2007} in the long run.

The agents in our computation framework derive estimates of cost and time.
Several approaches employ machine learning to predict the execution time of a
query~\cite{Ganapathi:2009, Akdere:2012}. Li et al. \cite{Li:2012} use statistical model to estimate CPU time and other resource comsumption. Duggan et al.~\cite{Duggan:2013, Duggan:2014} introduce special metrics and predict performance based on sampling. Wu et al. analyze
the query execution plan directly to derive runtime
predictions~\cite{Wu:2013:ICDE, Wu:2013:VLDB}, or use probabilistic
models~\cite{Wu:2014}.  
Ye et al.~\cite{Ye:2012} perform service composition given the resource requirement of individual tasks.
Uncertainty of time and cost is an important component in our framework.
Existing work on scheduling SLAs considers uncertainty in the completion time
when contracts specify a price.  
Specifically, scheduling considers 3
possible outcomes: (1) the provider accepts the SLA and returns results before
the deadline, earning some profit; (2) the provider accepts the SLA but misses
the deadline, and pays some penalty to the consumer; (3) the provider rejects
the SLA and pays some penalty immediately. 
Xiong et al.\ \cite{Xiong:2011}
have developed an SLA admission control system that predicts the distribution
of completion time, and accepts or rejects SLAs based on the expected profit.
Chi et al.\ \cite{Chi:2013} assume a stream of SLAs all of which must be accepted.
Their system minimizes loss by determining the execution order of SLAs
based on the uncertain completion time and the penalty of missing the deadline. 
Liu et al.\ \cite{Liu:2013} have proposed an algorithm to solve tenant placement in the cloud given the distribution of completion time and the penalty of SLAs.
Our market works differently in two aspects. First, our contracts consist of multiple target times, which are more flexible than the single deadline implicit in these SLAs. Second, we do not require the consumer to propose an SLA that may be rejected.  Instead, the consumer makes a request that is priced by the agent according to their capabilities.

Fine-grained contract pricing is related to query optimization in distributed databases~\cite{Lanzelotte:1993, Ganguly:1992} as we execute subtasks using different virtual machines. However, contract optimization has two objectives (time and cost), while query optimization has only one (time). These two objectives propagate differently in the task graph, making the problem more difficult.

%% file: discussion.tex

\section{Discussion} \label{sec:discussion}

Our framework can be easily extended to handle applications with
different QoS parameters. For example, in long-running services,
completion time is not relevant and thus should not be part of the
utility function. In contrast, other factors, such as response time,
are important. These parameters are also uncertain due to unstable
cloud performance~\cite{Lorenz:1998, Wang:2011}. While we did not
experiment with alternative QoS parameters and different application
settings, our market framework is already equipped to handle them with
appropriate changes to the utility function.

A meaningful extension to our work is to augment the market to handle
varying prices. Our current framework assumes fixed prices for
resource configurations. However, fluctuating prices do exist in the
real world. For example, Amazon EC2 allows agents bid spot instances
with much lower price~\cite{Yi:2010, Yehuda:2013}. Agents set a
maximum price threshold when requesting a spot instance. The request
can be fulfilled when the market price of a spot instance is lower
than this threshold. If the market price increases above the
threshold, the spot instance will be terminated. In addition, agents
can rent reserved instances either directly from Amazon EC2 or through
its Reserved Instance Marketplace. In these cases, agents have more
options to execute a task: 1) buy spot instances; 2) use their
previously reserved instances; 3) buy reserved instances from others.
These options introduce two additional factors to the market. First, the market needs to account for a supply function $S(\Rate, t)$. This means there are
$S(\Rate, t)$ instances with rate $\Rate$ and available time $t$,
where $\Rate$ is usually lower than the regular instance rate and $t$
must be a limited period ranging from hours to years. Second, the framework needs to consider the starting time of a task. The starting time is inconsequential when there are
only regular instances with fixed rates: The agent starts regular
instances whenever a consumer accepts the contract. However, the
starting time matters when the rates fluctuate. In this case, agents need to estimate
(a) the supply function at different points in time to ensure enough
machine hours for finishing consumers' tasks, and (b) the demand function
at different points in time to decide how many instances they want to reserve.
This is not a straightforward extension to our work, and will likely lead to a more complex market model.

%% file: appendix.tex

\appendix

\section{Linear case} \label{apx:linearCase}

We provide an analytical solution to the linear case, in which the utility function $U$ and demand function $M$ are linear:

\begin{itemize}[leftmargin=5mm]
\item The consumers specifies a linear utility function $\Util(t, \Price) =
-\ParaUtilT \cdot t - \ParaUtilP \cdot \Price$, which means they are always
willing to pay $\ParaUtilT$ units of cost to save $\ParaUtilP$ units of time.

\item The demand function is linear: $\Demand(\Util) = \ParaDemand +
\ParaDemandU \cdot \Util$, which means that when $\Util$ increased by
$1/\ParaDemandU$, $1$ more contract would be accepted. Since $\Util(t,
\Price)$ is linear, the demand function can be written as $\Demand(\Util) =
\ParaDemand - \ParaDemandT t - \ParaDemandP \Price$.

\end{itemize}

Applying Equations~\ref{eq:discreteProfit} and~\ref{eq:discreteDemand} to Problem~\ref{prob:Pricing}, we compute the overall profit as:
\begin{equation*}
\begin{aligned}
\ProfitAll 
= & \sum_{i=1}^{|I|} (\Price_i - c_i)p_i  \cdot  \sum_{i=1}^{|I|} \Demand \left( \Util(\hat{t}_i, \Price_i)\right) p_i \\
= & (\Price - c)^Tp  \cdot   \left(\ParaDemand - \ParaDemandT \expect^T p - \ParaDemandP \Price^T p \right) \\
= &  -\ParaDemandP \left( \Price^Tp - \frac{\ParaDemand - \ParaDemandT \expect^Tp + \ParaDemandP c^Tp}{2\ParaDemandP} \right)^2 \\
& + \frac{(\ParaDemand - \ParaDemandT \expect^Tp - \ParaDemandP c^Tp)^2}{4 \ParaDemandP} \\
\end{aligned}
\end{equation*}

$\ProfitAll$ is maximized when 
{\small
\begin{equation*}
\Price^Tp =
\left\{
\begin{aligned}
\frac{\ParaDemand - \ParaDemandT \expect^Tp + \ParaDemandP c^Tp}{2\ParaDemandP},\; &  \ParaDemand - \ParaDemandT \expect^Tp - \ParaDemandP c^Tp \geq 0 \\
c^Tp + \epsilon,\; & otherwise
\end{aligned}
\right.
\end{equation*}
}
where $\epsilon$ is a small positive value.

Furthermore, when $\ParaDemand - \ParaDemandT \expect^Tp - \ParaDemandP c^Tp \geq 0$, 
\begin{equation} \label{eq:appendixLinearProfit}
\ProfitAll =  \frac{(\ParaDemand - \ParaDemandT \expect^Tp - \ParaDemandP c^Tp)^2}{4 \ParaDemandP}
\end{equation}

\begin{example} 
\label{ex:oneBucket}
Let $\Itv = {[0, \infty)}$; then $p_1 = 1$, and $\hat{t}_1$ is the expected completion time. 
The only variable is $\Price_1$.  The problem becomes:
\begin{equation*}
\begin{aligned}
\mbox{maximize}: & \\
\ProfitAll = &  (\Price_1 - c_1)  \cdot (\ParaDemand - \ParaDemandT \hat{t}_1 - \ParaDemandP \Price_1) \\
\mbox{subject to}: & \\
E[profit] & = (\Price_1 - c_1) > 0\\ 
\end{aligned}
\end{equation*}
The result for the price $\pi_1$ is:
\begin{equation*}
\begin{aligned}
\Price_1 & = max\left(\frac{\ParaDemand - \ParaDemandT \expect_1 + \ParaDemandP c_1}{2\ParaDemandP}, c_1 + \epsilon \right) \\
\end{aligned}
\end{equation*}
\end{example}

We denote a configuration by its probability, time, and cost tuple: $(p, \expect, c)$. 
We can apply utility function directly to this configuration as $E[\Util(p, \expect, c)] = -\ParaUtilT \expect^Tp - \ParaUtilP \Price^Tp$.
We define that a configuration $(p_1, \expect_1, c_1)$ is better than another configuration $(p_2, \expect_2, c_2)$ when $E[\Util(p_1, \expect_1, c_1)] > E[\Util(p_2, \expect_2, c_2)]$.
If an agent finds a better configuration than another configuration, she/he can provide consumers greater utility (defined in Equation \ref{eq:expectedUtility}) while making more profit (defined in Equation \ref{eq:linearProfit}). Formally:

\begin{thm}
When $E[\Util(p_1, \expect_1, c_1)] > E[\Util(p_2, \expect_2, c_2)]$, consumer's utility $E[\Util_1] > E[\Util_2]$ and overall profit $\ProfitAll_1 > \ProfitAll_2$.
\end{thm}
\begin{proof}
We ignore the corner case in which $\ParaDemand - \ParaDemandT \expect^Tp - \ParaDemandP c^Tp < 0$, which means even the most efficient configuration found by the agent leads to zero demand. When $\ParaDemand - \ParaDemandT \expect^Tp - \ParaDemandP c^Tp \geq 0$, according to Section \ref{sec:linearCaseSolution}, the overall profit $\ProfitAll$ is maximized when $\Price^Tp = \frac{\ParaDemand - \ParaDemandT \expect^Tp + \ParaDemandP c^Tp}{2\ParaDemandP}$. So
\begin{equation*}
\begin{aligned}
E[\Util] = & -\ParaUtilT \expect^Tp - \ParaUtilP \Price^Tp \\
= & -\ParaUtilT \expect^Tp - \ParaUtilP \frac{\ParaDemand - \ParaDemandT \expect^Tp + \ParaDemandP c^Tp}{2\ParaDemandP} \\
= & -\ParaUtilT \expect^Tp - \ParaUtilP \frac{\ParaDemand - \ParaDemandU \ParaUtilT \expect^Tp + \ParaDemandU \ParaUtilP c^Tp}{2  \ParaDemandU \ParaUtilP } \\
= & - \frac{\ParaDemand}{2  \ParaDemandU} - \frac{\ParaUtilT \expect^Tp}{2} - \frac{ \ParaUtilP c^Tp}{2} \\
= & - \frac{\ParaDemand}{2  \ParaDemandU} + \frac{1}{2} E[\Util(p, \expect, c)]
\end{aligned}
\end{equation*}
So $E[\Util_1] > E[\Util_2]$ if $E[\Util(p_1, \expect_1, c_1)] > E[\Util(p_2, \expect_2, c_2)]$.

In addition, 
\begin{equation*}
\begin{aligned}
\ProfitAll = & \frac{(\ParaDemand - \ParaDemandT \expect^Tp - \ParaDemandP c^Tp)^2}{4 \ParaDemandP} \\
= & \frac{(\ParaDemand - \ParaDemandU  \ParaUtilT \expect^Tp - \ParaDemandU  \ParaUtilP c^Tp)^2}{4 \ParaDemandP} \\
= & \frac{(\ParaDemand + \ParaDemandU E[\Util(p, \expect, c)])^2}{4 \ParaDemandP}
\end{aligned}
\end{equation*}
Since $\ParaDemand - \ParaDemandT \expect^Tp - \ParaDemandP c^Tp = \ParaDemand + \ParaDemandU E[\Util(p, \expect, c)] \geq 0$, $\ProfitAll_1 > \ProfitAll_2$ if $E[\Util(p_1, \expect_1, c_1)] > E[\Util(p_2, \expect_2, c_2)]$. 
\end{proof}

\section{Fine-grained pricing} \label{apx:fineGrained}

\begin{thm}
Contract optimization in fine-grained pricing is NP-hard.  
\end{thm}
\begin{proof}
We prove this by reducing a 0-1 Knapsack problem to it.
In a knapsack problem, there are $n$ items, each of which has a value $v_i$ and a weight $w_i$. The maximum weight of a knapsack is $W$. One wants to maximize the sum of values $V$ of selected items while the sum of their weights does not exceed $W$. We construct a contract optimization problem correspondingly. We make $n$ subtasks in a chain. Their time and cost are deterministic. The $i$th subtask has two options $(w_i, v_0 - v_i)$ and $(0, v_0)$, where $v_0=\max_i\{v_i\}$. Then let the overall profit be 
$$
\ProfitAll(T, C) = 
\begin{cases}
n*v_0 - C & T \leq W\\
0 & T > W
\end{cases}$$
So one can achieve value $V$ in the knapsack problem without exceeding weight limit $W$ if and only if the $\ProfitAll = V$ in the contract optimization problem.
\end{proof}

\section{VCG-auction-based approach} \label{apx:vcg}

Recall the definition of $\Omega(t)$: When $\Util(t,\Price)= -\ParaUtilT t - \ParaUtilP \Price$ is linear, let $\Delta=\Util_i- \Util^*$. The inverse function of $\Util(t,\Price)$ is $\Pi(t,u)=(-\ParaUtilT t - u)/\ParaUtilP$. We define $\Omega(t)=\Pi(t,u-\Delta)=\Pi(t,u)+\Delta/\ParaUtilP$.

\begin{thm} \label{thm:vcgPayment}
When $\Util(t,\Price)$ is linear, the above $\Omega(t)$ satisfies:

1) $\Util_i > \Util^* > \Util^C_i \Rightarrow E[\text{\emph{profit}}] < 0$;

2) $\Util^C_i > \Util^* > \Util_i \Rightarrow E[\text{\emph{profit}}] > 0$.
\end{thm}

\begin{proof} ~

1) $\Delta=\Util_i-\Util^*$ and $\Util_i > \Util^* > \Util_i^C$ \\
$\Rightarrow \Delta<\Util_i-\Util_i^C=\sum_{k=1}^n p_k [\Util(\hat{t_k}, \Price_k) - \Util(\hat{t_k},c_k)]$\\
$\Rightarrow \sum_{k=1}^n p_k [\Util(\hat{t_k}, \Price_k) - \Util(\hat{t_k},c_k)] - \Delta > 0$. \\
$0\leq p_k \leq 1$ and $\sum_{k=1}^n p_k = 1$\\
$\Rightarrow \Delta=\sum_{k=1}^n p_k \Delta$.\\
So $\sum_{k=1}^n p_k [\Util(\hat{t_k}, \Price_k) - \Delta - \Util(\hat{t_k},c_k)] > 0$ \\
$\Rightarrow \sum_{k=1}^n p_k (-\ParaUtilP \Price_k - \Delta + \ParaUtilP c_k) > 0$.\\
So $E[\text{\emph{profit}}] = \sum_{k=1}^n p_k (\omega_k(\hat{t_k}) - \Cost_i(\hat{t_k})) = -\sum_{k=1}^n p_k (-\ParaUtilP \Price_k - \Delta + \ParaUtilP c_k) < 0$.

Therefore $\Util_i > \Util^* > \Util^C_i \Rightarrow E[\text{\emph{profit}}] < 0$;

2) Similar to above, we can prove $\Util^C_i > \Util^* > \Util_i \Rightarrow E[\text{\emph{profit}}] > 0$. 
\end{proof}

\begin{thm} 
Every agent truthfully revealing its cost is a weakly-dominant strategy. 
\end{thm}

\begin{proof}
Truthfully bidding means $\Util_i = \Util^C_i$.

1) The strategy of overbidding, $\Util_i > \Util_i^C$, is dominated by truthfully bidding. 

When $\Util^* > \Util_i > \Util_i^C$, both strategies yield $\text{\emph{payoff}}_i=0$. 

When $\Util_i > \Util_i^C > \Util^*$, both strategies yield $\text{\emph{payoff}}_i=E[\text{\emph{profit}}]>0$. 

However, when $\Util_i > \Util^* > \Util_i^C$, overbidding yields $\text{\emph{payoff}}_i=E[\text{\emph{profit}}]<0$ (Theorem~\ref{thm:vcgPayment}) while truthfully bidding yields $\text{\emph{payoff}}_i=0$. 

So overbidding is dominated by truthfully bidding.

2) The strategy of underbidding, $\Util_i < \Util_i^C$, is dominated by truthfully bidding.

When $\Util^*>\Util^C_i>\Util_i$, both strategies yield $\text{\emph{payoff}}_i=0$. 

When $\Util_i^C > \Util_i > \Util^*$, both strategies yield $\text{\emph{payoff}}_i=E[\text{\emph{profit}}]>0$. 

However, when $\Util_i^C > \Util^* > \Util_i$, underbidding yields $\text{\emph{payoff}}_i=0$ while truthfully bidding yields $\text{\emph{payoff}}_i=E[\text{\emph{profit}}]>0$ (Theorem~\ref{thm:vcgPayment}). 

So underbidding is dominated by truthfully bidding.

So truthfully bidding is a weakly-dominant strategy.
\end{proof}

\section{Bertrand Model} \label{apx:bertrand}

Recall that in the Differentiated Bertrand Model with nonidentical demand functions, the best response functions are given by
\begin{align} 
\mu_1=f_1(\mu_2) =  a_1\mu_2 + b_1,  \\
\mu_2=f_2(\mu_1)= a_2 \mu_1 + b_2,
\end{align}
where
\[
a_i = \frac{\beta_i}{2\alpha_i}, \quad b_i = \frac{\gamma_i}{2\alpha_i},\quad \text{for } i=1,2.
\]
If $a_1 a_2 < 1$, there exists a Nash Equilibrium $(\mu_1^*, \mu_2^*)$, which is the unique solution to the following equations,
\[
\begin{cases}
\mu_1^* = f_1(\mu_2^*),\\
\mu_2^* = f_2(\mu_1^*).
\end{cases}
\]

Suppose that both agents keep updating their prices to the best response to the currently observed price of the other party. We show that their prices eventually converges to the NE $(\mu_1^*, \mu_2^*)$.

We allow for asynchronous updates. Thus at the $x$-th update, there are three possibilities,
\begin{enumerate}
\item[(1)] Only agent 1 updates his prices, in which case
\begin{equation}
(\mu_{1,x+1},\mu_{2,x+1})  = F_1(\mu_{1,x},\mu_{2,x}) = (f_1(\mu_{2,x}), \mu_{2,x}). \label{eq:f1}
\end{equation}
\item[(2)] Only agent 2 updates his prices, in which case
\begin{equation}
(\mu_{1,x+1},\mu_{2,x+1}) = F_2(\mu_{1,x},\mu_{2,x})=(\mu_{1,x}, f_2(\mu_{1,x})). \label{eq:f2}
\end{equation}
\item[(3)]  Both agents update their prices, in which case
\begin{equation}
(\mu_{1,x+1},\mu_{2,x+1})  = F_3(\mu_{1,x},\mu_{2,x}) = (f_1(\mu_{2,x}), f_2(\mu_{1,x})). \label{eq:both}
\end{equation}
\end{enumerate}
Thus $(\mu_{1,x}, \mu_{2,x}) = G_x\circ G_{x-1} \circ\dots\circ G_1(\mu_{1,0}, \mu_{2,0})$, where $\mu_{1,0}$ and $\mu_{2,0}$ are the initial prices, and $G_i \in \{F_1, F_2, F_3\}$ for $i=1,2,\dots,x$.
We assume that both agents keep updating their prices, i.e., 
\begin{equation}\label{eq:diverge}
\lim_{x\to\infty} \sum_{i=1}^x \mathbf{1}\{G_i \neq F_1\} = \lim_{x\to\infty}  \sum_{i=1}^x \mathbf{1}\{G_i \neq F_2\} = \infty. 
\end{equation}

\begin{thm}
Assume $a_1 a_2 < 1$. If both agents keep updating their prices, i.e. \eqref{eq:diverge} holds, then for any initial prices $\mu_{1,0}$ and $\mu_{2,0}$, we have
\[
\lim_{x\to\infty} (\mu_{1,x}, \mu_{2,x}) = (\mu_1^*, \mu_2^*).
\]
\end{thm}

\begin{proof} For any function $f$, let $f^{(0)} = \text{id}$, the identity function, and $f^{(m)} = f^{(m-1)}\circ f$ for $m \geq 1$. 
Note that $F_i^{(2)} = F_i$ and $F_3 \circ F_i = F_{3-i}\circ F_i$ for $i=1,2$. By repeated applications of these relations, we obtain  
\begin{equation}\label{eq:iteration}
(\mu_{1,x}, \mu_{2,x}) =  F_2^{(m_2)}\circ H_1^{(m)} \circ F_1^{(m_1)} \circ F_3^{(m_3)}(\mu_{1,0}, \mu_{2,0}),
\end{equation}
where $H_1 = F_1\circ F_2$, $m_1\in \{0,1\}$, $m_2 \in \{0,1\}$ and $m_1 + m_2 + m_3 + m = N_x$. Note that $N_x$ is the number of \emph{effective} updates that can potentially change the prices. Without loss of generality, we assume that $N_x= x$, which amounts to discarding those updates that cannot change the price of either agent. 

\itemize{
\item Case I: $m_1 = m_2 = m = 0$

In this case, both agents always synchronize their updates and
\begin{align*}
& (\mu_{1,x}, \mu_{2,x}) = F_3^{(x)} (\mu_{1,0}, \mu_{2,0}) \\
& = \begin{cases}
(g_1^{(\frac{x}{2})}(\mu_{1,0}), g_2^{(\frac{x}{2})} (\mu_{2,0})), & \text{for $x$ even},\\
(g_1^{(\frac{x-1}{2})}\circ f_1(\mu_{2,0}), g_2^{(\frac{x-1}{2})} \circ f_2 (\mu_{1,0})), & \text{for $x$ odd},
\end{cases}
\end{align*}
where $g_1 = f_1\circ f_2$ and $g_2 = f_2\circ f_1$. Note that $g_1(\mu_1^*) = \mu_1^*$, and, for any $z$,
\[
g_1(z) - \mu_1^*  = g_1(z) - g_1(\mu_1^*) = a_1[f_2(z) - f_2(\mu_1^*)] = a_1 a_2 (z-\mu_1^*).
\]
Thus for any $z$,
\[
|g_1^{(k)}(z) - \mu_1^*| = (a_1 a_2)^k |z-\mu_1^*| \to 0, \quad \text{as } k\to \infty.
\]
Similarly, for any $z$,
\begin{equation}\label{eq:g2}
|g_2^{(k)}(z) - \mu_2^*| = (a_1 a_2)^k |z-\mu_2^*| \to 0, \quad \text{as } k\to \infty.
\end{equation}
It follows that
\[
\lim_{x\to\infty} (\mu_{1,x}, \mu_{2,x}) = (\mu_1^*, \mu_2^*).
\]

\item Case II: $m_1 + m_2 + m > 0$

In this case, the agents do not always synchronize their updates and $m\to\infty$ in \eqref{eq:iteration}. By symmetry, we assume $m_1 = 1$; the other case can be dealt with similarly. Note that
\[
H_1^{(m)}\circ F_1 (\mu_1, \mu_2) = (f_1\circ g_2^{(\lfloor \frac{m}{2}\rfloor)}(\mu_2), g_2^{(\lceil\frac{m}{2}\rceil)} (\mu_2)).
\]
Let $(\mu_1,\mu_2) = F_3^{(m_3)}(\mu_{1,0}, \mu_{2,0})$. By \eqref{eq:g2},
\[
\lim_{m\to\infty} H_1^{(m)}\circ F_1 (\mu_1, \mu_2) = (f_1(\mu_2^*), \mu_2^*) = (\mu_1^*, \mu_2^*).
\]
It follows that 
\begin{align*}
& \lim_{x\to\infty} (\mu_{1,n}, \mu_{2,n}) = \lim_{m\to\infty} F_2^{(m_2)} \circ H_1^{(m)}\circ F_1 (\mu_1, \mu_2) \\ 
& = F_2^{(m_2)} (\mu_1^*, \mu_2^*) = (\mu_1^*, \mu_2^*),
\end{align*}
where in the last step we have used the fact that $(\mu_1^*, \mu_2^*)$ is a fixed point of $F_2^{(m_2)}$ for $m_2 = 0,1$.
}
\end{proof}